\documentclass[11pt,draftcls,onecolumn]{IEEEtran}
\usepackage{cite}

\usepackage{array,multirow}
\usepackage{graphicx}  % Written by David Carlisle and Sebastian Rahtz
\usepackage{subfigure,color} % Written by Steven Douglas Cochran
\usepackage{amsmath}   % From the American Mathematical Society
\usepackage{amssymb}   % A popular package that provides many helpful commands
\usepackage{eucal}     % Calligraphic letters

\usepackage{pxfonts}
\usepackage{dsfont}

\newtheorem{theorem}{\textbf{Theorem}}
\newtheorem{corollary}{\textbf{Corollary}}
\newtheorem{proposition}{\textbf{Proposition}}
\newtheorem{definition}{\textbf{Definition}}
\newtheorem{lemma}{Lemma}
\newtheorem{remark}{Remark}
\newcommand{\dv}{\mathbf} % determenistic vector
\newcommand{\mc}{\mathcal} % determenistic vector
\newcommand{\qed}{\hfill \ensuremath{\Box}}

\DeclareMathAlphabet{\eurm}{U}{eur}{m}{n}
\DeclareMathAlphabet{\mathbsf}{OT1}{cmss}{bx}{n}% bold sans serif
\DeclareMathAlphabet{\mathssf}{OT1}{cmss}{m}{sl}% slanted sans serif
\DeclareMathAlphabet{\mathcsf}{OT1}{cmss}{sbc}{n}% condensed sans serif
\DeclareSymbolFont{bsfletters}{OT1}{cmss}{bx}{n}  
\DeclareSymbolFont{ssfletters}{OT1}{cmss}{m}{n}
\DeclareMathSymbol{\bsfGamma}{0}{bsfletters}{'000}
\DeclareMathSymbol{\ssfGamma}{0}{ssfletters}{'000}
\DeclareMathSymbol{\bsfDelta}{0}{bsfletters}{'001}
\DeclareMathSymbol{\ssfDelta}{0}{ssfletters}{'001}
\DeclareMathSymbol{\bsfTheta}{0}{bsfletters}{'002}
\DeclareMathSymbol{\ssfTheta}{0}{ssfletters}{'002}
\DeclareMathSymbol{\bsfLambda}{0}{bsfletters}{'003}
\DeclareMathSymbol{\ssfLambda}{0}{ssfletters}{'003}
\DeclareMathSymbol{\bsfXi}{0}{bsfletters}{'004}
\DeclareMathSymbol{\ssfXi}{0}{ssfletters}{'004}
\DeclareMathSymbol{\bsfPi}{0}{bsfletters}{'005}
\DeclareMathSymbol{\ssfPi}{0}{ssfletters}{'005}
\DeclareMathSymbol{\bsfSigma}{0}{bsfletters}{'006}
\DeclareMathSymbol{\ssfSigma}{0}{ssfletters}{'006}
\DeclareMathSymbol{\bsfUpsilon}{0}{bsfletters}{'007}
\DeclareMathSymbol{\ssfUpsilon}{0}{ssfletters}{'007}
\DeclareMathSymbol{\bsfPhi}{0}{bsfletters}{'010}
\DeclareMathSymbol{\ssfPhi}{0}{ssfletters}{'010}
\DeclareMathSymbol{\bsfPsi}{0}{bsfletters}{'011}
\DeclareMathSymbol{\ssfPsi}{0}{ssfletters}{'011}
\DeclareMathSymbol{\bsfOmega}{0}{bsfletters}{'012}
\DeclareMathSymbol{\ssfOmega}{0}{ssfletters}{'012}

\begin{document}

\fontencoding{OT1}\fontsize{9.4}{11.25pt}\selectfont
\title{Capacity Region of Multiple Access Channel with States Known Noncausally at One Encoder and Only Strictly Causally at the Other Encoder}
\vspace{1cm}

\author{\vspace{0cm}
\authorblockN{ \small Abdellatif Zaidi \qquad Pablo Piantanida \qquad Shlomo Shamai (Shitz)\thanks{The material in this paper was presented in part at the IEEE International Symposium on Information Theory, Saint-Petersburg, Russia, August 2011. This work has been supported by the European Commission in the framework of the FP7 Network of Excellence in Wireless Communications (NEWCOM++). The work of S. Shamai has also been supported by the CORNET consortium.}
\thanks{Abdellatif Zaidi is with Universit\'e Paris-Est Marne La Vall\'ee, 77454 Marne la Vall\'ee Cedex 2, France. Email: abdellatif.zaidi@univ-mlv.fr}
\thanks{Pablo Piantanida is with the Department of Telecommunications, SUPELEC, 91190 Gif-sur-Yvette, France. Email: pablo.piantanida@supelec.fr}
\thanks{Shlomo Shamai is with the Department of Electrical Engineering, Technion Institute of Technology, Technion City, Haifa 32000, Israel. Email: sshlomo@ee.technion.ac.il}}}

\vspace{1cm}

% make the title area
\maketitle

\begin{abstract}
We consider a two-user state-dependent multiaccess channel in which the states of the channel are known non-causally to one of the encoders and only strictly causally to the other encoder. Both encoders transmit a common message and, in addition, the encoder that knows the states non-causally transmits an individual message. We find explicit characterizations of the capacity region of this communication model in both discrete memoryless and memoryless Gaussian cases. In particular the capacity region analysis demonstrates the utility of the knowledge of the states only strictly causally at the encoder that sends only the common message in general. More specifically, in the discrete memoryless setting we show that such a knowledge is beneficial and increases the capacity region in general. In the Gaussian setting, we show that such a knowledge does not help, and the capacity is same as if the states were completely unknown at the encoder that sends only the common message. Furthermore, we also study the special case in which the two encoders transmit only the common message and show that the knowledge of the states only strictly causally at the encoder that sends only the common message is not beneficial in this case, in both discrete memoryless and memoryless Gaussian settings. The analysis also reveals optimal ways of exploiting the knowledge of the state only strictly causally at the encoder that sends only the common message when such a knowledge is beneficial. The encoders collaborate to convey to the decoder a lossy version of the state, in addition to transmitting the information messages through a generalized Gel'fand-Pinsker binning. Particularly important in this problem are the questions of 1) optimal ways of performing the state compression and 2) whether or not the compression indices should be decoded uniquely. By developing two optimal coding schemes that perform this state compression differently, we show that when used as parts of appropriately tuned encoding and decoding processes, both compression \`a-la noisy network coding, i.e., with no binning, and compression using Wyner-Ziv binning are optimal. The scheme that uses Wyner-Ziv binning shares elements with Cover and El Gamal original compress-and-forward, but differs from it mainly in that backward decoding is employed instead of forward decoding and the compression indices are not decoded uniquely. Finally, by exploring the properties of our outer bound, we show that, although not required in general, the compression indices can in fact be decoded uniquely essentially without altering the capacity region, but at the expense of larger alphabets sizes for the auxiliary random variables.

%The capacity region analysis also demonstrates the utility of the knowledge of the states only strictly causally at the encoder that sends only the common message in general. More specifically, in the discrete memoryless setting we show that such a knowledge is beneficial and increases the capacity region in general. In the Gaussian setting, we show that such a knowledge does not help, and the capacity is same as if the states were completely unknown at the encoder that sends only the common message. Furthermore, we also study the special case in which the two encoders transmit only the common message and show that the knowledge of the states only strictly causally at the encoder that sends only the common message is not beneficial in this case, in both discrete memoryless and memoryless Gaussian settings. 

% The coding schemes are based on a block Markov scheme in which a lossy version of the state is conveyed to the decoder, combined with a generalized Gel'fand-Pinsker binning for the transmission of the information messages. Investigating ways of performing the state compression, we show that both compression \`a la Kim et al. noisy network coding with no binning, repetitive encoding and joint unique decoding and compression using Wyner-Ziv binning and backward decoding are optimal.
\end{abstract}

\section{Introduction}\label{secI}

The study of channels that are controlled by random states has spurred much interest, due to its importance from both information-theoretic and communications aspects. For example, state-dependent channels may model communication in random fading environments \cite{BPS98} or in the presence of interference imposed by adjacent users. The channel states may be known in a strictly-causal, causal or noncausal manner, to all or only a subset of the encoders. For a transmission of length $n$, let $S^n=(S_1,S_2,\hdots,S_n)$ denote the state sequence, with $S_i$ representing the channel state affecting the channel at time or block $i$. For the transmission in block $i$, the state sequence is known non-causally if it is known entirely before the beginning of the transmission. It is known causally if it is known up to and including time $i$; and it is known strictly causally if it is known only up to time $i-1$. The way the channel state information is utilized and influences capacity depends also on which of the encoders(s) and decoder(s) are aware of it. In single user channels, the concept of channel state available at only the transmitter dates back to Shannon \cite{Sh58} for the causal channel state case, and to  Gel'fand and Pinsker \cite{GP80} for the non-causal channel state case. In multiuser environments, a growing body of work studies multi-user state-dependent models. Recent advances in this regard can be found in \cite{KSM08,SBSV07a,KL07a,ZKLV09a,KELW07,PKEZ07,ZKLV08a,ZKLV10,ZV09b,CS05,S05,LS10a,LS10b,PSS11,LSY10,LSY11,AMA09,KS09a,CY11a,SCYA11a,K-FM11a,BL10,J06,SK05}, and many other works. For a comprehensive review of state-dependent channels and related work, the reader may refer to \cite{KSM08}.

There is a connection between the role of states known strictly causally at an encoder and that of output feedback given to that encoder. In single-user channels, it is now well known that strictly causal feedback does not increase the capacity \cite{Sh56}. In multiuser channels or networks, however, the situation changes drastically, and output feedback can be beneficial --- but its role is still highly missunderstood. One has a similar picture with strictly causal states at the encoder. In single-user channels, independent and identically distributed states available only in a strictly causal manner at the encoder have no effect on the capacity. In multiuser channels or networks, however, like feedback, strictly causal states in general increase the capacity. 

Advances in the study of the effect of strictly causal states in multiuser channels are rather very recent and concern mainly multiple access scenarios. In \cite{LS10a}, Lapidoth and Steinberg study a two-encoder multiple access channel with independent messages and states known causally at the encoders. They show that the strictly causal state sequence can be beneficial, in the sense that it increases the capacity for this model. This result is reminiscent of Dueck's proof \cite{D80} that feedback can increase the capacity region of some broadcast channels. In accordance with \cite{D80}, the main idea of the achievability result in \cite{LS10a} is a block Markov coding scheme in which the two users collaborate to describe the state to the decoder by sending cooperatively a compressed version of it. As noticed in \cite{LS10a}, although some non-zero rate that otherwise could be used to transmit pure information is spent in describing the state to the decoder, the net effect can be an increase in the capacity. In \cite{LS10b}, they show that strictly causal state information is beneficial even if the channel is controlled by two independent states each known to one encoder strictly causally. In this case, each encoder can help the other encoder transmit at a higher rate by sending a compressed version of its state to the decoder. In \cite{LSY10}, Li, Simeone and Yener improve the results of \cite{LS10a,LS10b} and extend them to the case of multiple encoders.  The achievability results in \cite{LSY10} are inspired by the noisy network coding scheme of \cite{H-LKGC11} and, unlike \cite{LS10a,LS10b}, do not use Wyner-Ziv binning \cite{WZ76} for the compression of the state. In a very recent contribution \cite{LS11a}, Lapidoth and Steinberg derive a new inner bound on the capacity region for the case of a single state governing the multiaccess channel. They also prove that the inner bound of \cite{LSY10} for the case of two independent states each known strictly causally to one encoder can indeed be strictly better than previous bounds in \cite{LS10a,LS10b} -- a result which is conjectured previously by Li, Simeone and Yener in \cite{LSY10}.

% a conjecture previously made by Li, Simeone and Yener in \cite{LSY10} that their inner bound for the case of two independent states each known strictly causally to one encoder can indeed be strictly better than previous bounds.

%Multiple access channels with states known causally at the encoders  have been studied recently in \cite{LS10a, LS10b} and \cite{LSY10} (see also \cite{BL10,J06,SK05}). In \cite{LS10a}, the states are known in a strictly causal manner at both encoders which transmit independent messages. The authors show that the strict knowledge of the states can be beneficial, in the sense that it increases the capacity for this model. This result is reminiscent of Dueck's proof \cite{D80} that feedback can increase the capacity region of some broadcast channels. In accordance with \cite{D80}, the main idea of the achievability result in \cite{LS10a} is a block Markov coding scheme in which the two users collaborate to describe the state to the decoder by sending cooperatively a compressed version of it. As noticed in \cite{LS10a}, although some non-zero rate that otherwise could be used to transmit pure information is spent in describing the state to the decoder, the net effect can be an increase in the capacity.

\subsection{Studied Model}
%In \cite{LS10a} and \cite{LS10b}, an encoder that benefits from the availability of states at the other encoder (strictly causally) does not know the states fully itself, i.e., it knows the states only strictly causally itself. One can then wonder whether, in a multiaccess channel, the knowledge of the states only strictly causally at one encoder could be of any help to another encoder which knows the states non-causally. 

In this paper, which generalizes a former conference version \cite{ZPS11a}, we study a two-user state-dependent multiple access channel with the channel states known non-causally at one encoder and only strictly causally at the other encoder. The decoder is not aware of the channel states. As shown in Figure~\ref{ModelForMACwithAsymmetricCSI}, both encoders transmit a common message and, in addition, the encoder that knows the states non-causally transmits an individual message. This model generalizes one whose capacity region is established in \cite{SBSV07a} and in which the encoder that sends only the common message does not know the states at all. More precisely, let $W_c$ and $W_1$ denote the common message and the individual message to be transmitted in, say, $n$ uses of the channel; and $S^n=(S_1,\hdots,S_n)$ denote the state sequence affecting the channel during this time. At time $i$, Encoder 1 knows the complete sequence $S^n=(S_1,\hdots,S_{i-1},S_i,\hdots,S_n)$ and sends $X_{1i}=\phi_1(W_c,W_1,S^n)$, and Encoder 2 knows \textit{only} $S^{i-1}=(S_1,\hdots,S_{i-1})$ and sends $X_{2i}=\phi_{2,i}(W_c,S^{i-1})$ -- the functions $\phi_1$ and $\phi_{2,i}$ are some encoding functions. In this paper, we study the capacity region of this state-dependent MAC model. As our analysis will show, this requires, among others, understanding the role of the strictly causal part of the state that is revealed to Encoder 2.

\begin{figure}[htpb]
\centering
\includegraphics[width=0.7\linewidth]{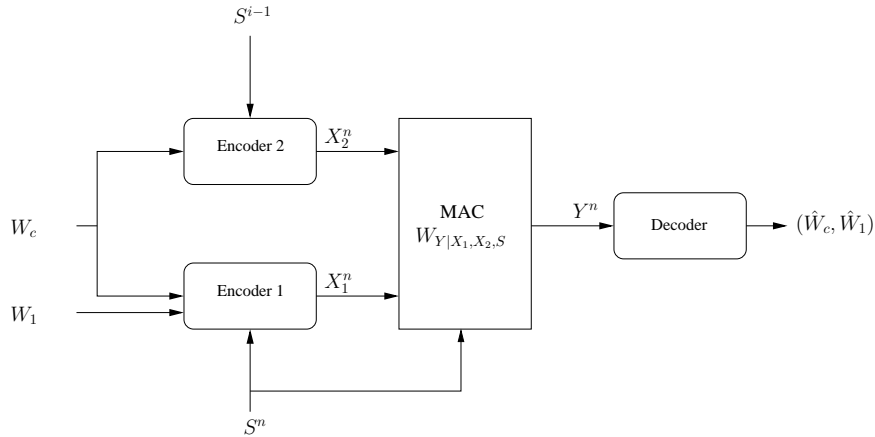}
\caption{State-dependent MAC with degraded message sets and states known noncausally at the encoder that sends both messages and only strictly causally at the other encoder.}
\label{ModelForMACwithAsymmetricCSI}
\end{figure}

\subsection{Main Contributions}

In the discrete memoryless case, we characterize the capacity region for the general finite-alphabet case with a single-letter expression. The proof of the achievability part is based on a block-Markov coding scheme in which the two encoders collaborate to convey a lossy version of the state to the decoder, in the spirit of \cite{LS10a,LS10b,LS11a}, in addition to a generalized Gel'fand-Pinsker binning for the transmission of the information messages \cite{GP80}. From the angle of the state compression, coding schemes that perform the state compression for our model tie with very recent works on compressions in compress-and-forward type relaying networks \cite{H-LKGC11,WX11a,ZHV11a,KH11a}. We first develop a coding scheme in which the state compression is performed \`a-la Kim \textit{et al.} noisy network coding scheme and show that it is optimal, i.e., achieves an outer bound that we establish for the studied model. In this coding scheme, unlike \cite{LS10a,LS10b,LS11a} where every information message is divided into blocks and different submessages are sent over these blocks and then decoded one at a time using the same codebook as in the original compress-and-forward scheme by Cover and El Gamal \cite{CG79}, here the \textit{entire} common message and the \textit{entire} individual message are transmitted over \textit{all} blocks using codebooks that are generated independently, one for each block, and the decoding is performed simultaneously using all blocks as in the noisy network coding scheme of \cite{H-LKGC11}. Also, like \cite{H-LKGC11}, at each block the compression index of the state of the previous block is sent using standard rate distortion, not Wyner-Ziv binning. At the end of the transmission, the receiver uses the outputs of all blocks to perform simultaneous decoding of the information common and individual messages, without uniquely decoding the compression indices. From this angle, our coding scheme connects more with \cite{LSY10}, than with \cite{LS10a,LS10b} and \cite{LS11a}. 

Two of the most important features of our coding scheme that is based on noisy network coding are i) standard compression without Wyner-Ziv binning and ii) non-explicit decoding of the compression indices. Investigating whether these features are pivotal for optimality in our problem, as argued in \cite{H-LKGC11} for some related models, we also explore binning-based compressions. We show that the capacity region of our model can also be achieved using an alternate coding scheme in which the state compression is realized using Wyner-Ziv binning. The employed optimal alternate coding scheme shares elements with Cover and El Gamal compress-and-forward \cite{CG79}, but differs from it in two aspects: 1) backward decoding is utilized instead of the forward decoding of \cite{CG79}, and 2) unlike \cite{CG79}, the compression indices are not decoded uniquely. Decoding backwardly instead of forwardly seems essential for the optimality of this alternate coding scheme here. At this level, we note that the finding in this paper that backward decoding with non-unique decoding of the compression indices is beneficial, may hold more generally in other scenarios that involve Wyner-Ziv binning. In the fading setting, this is also observed in \cite{KS09b}. Next, by exploring our outer bound further, we show that, although not required, one can modify this coding scheme in a manner to get the compression indices decoded at the receiver essentially without altering the capacity region but at the expense of larger alphabets sizes of the involved auxiliary random variables. The decoding of the compression indices introduces an additional rate constraint; but we show that this constraint is satisfied by the auxiliary random variables of the outer bound. Finally, we note that the finding in this paper that in the context of Wyner-Ziv binning backward decoding with non-unique decoding of the compression indices improves the transmission rate may be beneficial in other scenarios. In the fading setting, this was also observed in \cite{KS09b}.

The single-letter characterization of the capacity region of our model remains intact if one allows feedback to the encoder that sends both messages. Also, the capacity region of our model contains that of the model of \cite{SBSV07a} in which the encoder that sends only the common message is unaware of the channel states; and this shows that revealing the states even only strictly causally to this encoder potentially increases the capacity region. Next, by investigating a discrete memoryless example, we show that this inclusion can be strict, thus demonstrating the utility of conveying a compressed version of the state to the decoder cooperatively by the encoders.

We also specialize our results to the case in which the two encoders send only the common message. We refer to the capacity in this case as \textit{common-message} capacity. We show that, when one of the two encoders is informed noncausally, the knowledge of the states only strictly causally at the other encoder does not increase the common-message capacity. It should be noted that this result is not a direct consequence of that feedback does not increase the capacity in a multiaccess channel in which the encoders send only a common message; and our converse proof is needed here.

Next, we consider the memoryless Gaussian setting in which the channel state and the noise are additive and Gaussian. We establish an operative outer bound on the achievable rate pairs. Then, we show that this outer bound is achievable, yielding a closed-form expression of the capacity region. The resulting capacity region coincides with that of the model of \cite{SBSV07a} in which the encoder that sends only the common message is completely unaware of the states, thus demonstrating that, by opposition to the discrete memoryless case, revealing the states strictly causally to this encoder is not beneficial in the Gaussian case, in the sense that it does not increase the capacity region.

Finally, we note that in contrast to the related MAC models in \cite{SBSV07a,ZKLV09a}, our converse proofs in this paper do not follow directly from the converse part proof of the capacity formula for the standard Gel'fand-Pinsker channel \cite{GP80}. This is because, at time $i$, the encoder that transmits only the common message sends inputs which are function of not only that message, but also the observed past state sequence.

% Also, we show that our model has the \textit{same} capacity region as the one in \cite{SBSV07a} in the memoryless Gaussian case; and the \textit{same} common-message capacity as the one in \cite{SBSV07a} in both discrete memoryless and Gaussian cases. That is, the knowledge of the states strictly causally at Encoder 2 does \textit{not} increase the capacity; or, equivalently, Encoder 2 does no better than had it known only the message $W_c$.

% The converse proof for the model with only common message in this paper shows that the knowledge of the state in a strictly causal manner does not increase the capacity for this model. Also, the converse proof for the model with both common and individual messages shows that a coding scheme \`a-la Lapidoth-Steinberg (or any other) would not increase capacity in the Gaussian case for this model as well. For instance, this proves the non-utility of a joint description of the state $S_{i-1}$ to the decoder in the spirit of \cite{LS10a} for this model.

\subsection{Outline and Notation}

An outline of the remainder of this paper is as follows. Section \ref{secII} describes in more detail the communication model that we consider in this work. Section \ref{secIII} provides the capacity region of the discrete memoryless model. In this section we also establish an alternative outer bound on the capacity region that will turn to be useful in the Gaussian case, provide an example demonstrating the utility of revealing the states only strictly causally to the encoder that sends only the common message, and derive the common-message capacity. Section \ref{secIV} characterizes the capacity region as well as the common-message capacity of the Gaussian model. Finally, Section \ref{secV} concludes the paper.

We use the following notations throughout the paper. Upper case letters are used to denote random variables, e.g., $X$; lower case letters are used to denote realizations of random variables, e.g., $x$; and calligraphic letters designate alphabets, i.e., $\mc X$. The probability distribution of a random variable $X$ is denoted by $P_X(x)$. Sometimes, for convenience, we write it as $P_X$.  We use the notation $\mathbb{E}_{X}[\cdot]$ to denote the expectation of random variable $X$. A probability distribution of a random variable $Y$ given $X$ is denoted by $P_{Y|X}$. The set of probability distributions defined on an alphabet $\mc X$ is denoted by $\mc P(\mc X)$. The cardinality of a set $\mc X$ is denoted by $|\mc X|$. For convenience, the length $n$ vector $x^n$ will occasionally be denoted in boldface notation $\dv x$. The Gaussian distribution with mean $\mu$ and variance $\sigma^2$ is denoted by $\mathcal{N}(\mu,\sigma^2)$. For integers $i \leq j$, we define $[i:j]:=\{i,i+1,\hdots,j\}$. Finally, throughout the paper, logarithms are taken to base $2$, and the complement to unity of a scalar $u \in [0,1]$ is denoted by $\bar{u}$, i.e., $\bar{u}=1-u$.

%------------------------------------------ HERE ------------------------
\section{System Model and Definitions}\label{secII}

We consider a stationary memoryless state-dependent MAC $W_{Y|X_1,X_2,S}$  whose output $Y \in \mc Y$ is controlled by the channel inputs $X_1 \in \mc X_1$ and $X_2 \in \mc X_2$ from the encoders and the channel state $S \in \mc S$ which is drawn according to a memoryless probability law $Q_S$. We assume that the channel state $S^n$ is known non-causally at Encoder 1, i.e., beforehand, at the beginning of the transmission block. Encoder 2 knows the channel states only strictly-causally; that is, at time $i$, it knows the states only up to time $i-1$, $S^{i-1}=(S_1,\hdots,S_{i-1})$.

Encoder 2 wants to send a common message $W_c$ and Encoder 1 wants to send an independent individual message $W_1$ along with the common message $W_c$. We assume that the common message $W_c$ and the individual message $W_1$ are independent random variables drawn uniformly from the sets $\mc W_c=\{1,\cdots,M_c\}$ and  $\mc W_1=\{1,\cdots,M_1\}$, respectively. The sequences $X_{1}^n$ and $X_{2}^n$ from the encoders are sent across a state-dependent multiple access channel modeled as a memoryless conditional probability distribution $W_{Y|X_1,X_2,S}$. The joint probability mass function on ${\mc W_c}{\times}{\mc W_1}{\times}{\mc S^n}{\times}{\mc X^n_1}{\times}{\mc X^n_2}{\times}{\mc Y^n}$ is given by
\begin{align}
P(w_c,w_1,s^n,x^n_1,x^n_2,y^n) &= P(w_c)P(w_1)\prod_{i=1}^{n}Q_S(s_i)P(x_{1,i}|w_c,w_1,s^n)P(x_{2,i}|w_c,s^{i-1})\nonumber\\
&\hspace{1cm}{\cdot}W_{Y|X_1,X_2,S}(y_i|x_{1,i},x_{2,i},s_i).
\end{align}
The receiver guesses the pair $(\hat{W}_c,\hat{W}_1)$ from the channel output $Y^n$.

\begin{definition}
For positive integers $n$, $M_c$ and $M_1$, an $(M_c,M_1,n,\epsilon)$ code for the multiple access channel with states known noncausally at one encoder and only strictly causally at the other encoder consists of a mapping
\begin{align}
\phi_1: \mc W_c{\times}\mc W_1{\times}\mc S^n \longrightarrow \mc X^n_1
\label{EncodingFunction__Encoder1}
\end{align}
at Encoder 1, a sequence of mappings
\begin{align}
\phi_{2,i}: \mc W_c{\times}\mc S^{i-1} \longrightarrow \mc X_2, \quad i=1,\hdots,n
\label{EncodingFunction__Encoder2}
\end{align}
at Encoder 2, and a decoder map
\begin{align}
\psi : \mc Y^n \longrightarrow \mc W_c{\times}\mc W_1
\label{DecodingFunction}
\end{align}
such that the average probability of error is bounded by $\epsilon$,
\begin{equation}
P_e^n = \mathbb{E}_{S}\big[\mathrm{Pr}\big(\psi(Y^n)\neq (W_c,W_1)|S^n=s^n\big)\big] \leq \epsilon.
\end{equation}
The rate of the common message and the rate of the individual message are defined as
\begin{align}
&R_c = \frac{1}{n}\log M_c \qquad \text{and} \qquad R_1 = \frac{1}{n}\log M_1,
\end{align}
respectively.
\end{definition}

A rate pair $(R_c,R_1)$ is said to be achievable if for every $\epsilon > 0$ there exists an $(2^{nR_c},2^{nR_1},n,\epsilon)$ code for the channel $W_{Y|X_1,X_2,S}$.  The capacity region of the considered state-dependent MAC is defined as the closure of the set of achievable rate pairs.

\section{Discrete Memoryless Case}\label{secIII}

In this section, it is assumed that the alphabets $\mc S, \mc X_1, \mc X_2$ are finite.

\subsection{Capacity Region}\label{secIII_subsecA}
Let $\mc P$ stand for the collection of all random variables $(S,U,V,X_1,X_2,Y)$ such that $U$, $V$, $X_1$ and $X_2$ take values in finite alphabets $\mc U$, $\mc V$, $\mc X_1$ and $\mc X_2$, respectively, and
\begin{subequations}
\begin{align}
P_{S,U,V,X_1,X_2,Y}(s,u,v,x_1,x_2,y) &= P_{S,U,V,X_1X_2}(s,u,v,x_1,x_2)W_{Y|X_1,X_2,S}(y|x_1,x_2,s)\\
P_{S,U,V,X_1,X_2}(s,u,v,x_1,x_2) &= Q_S(s)P_{X_2}(x_2)P_{V|S,X_2}(v|s,x_2)P_{U,X_1|S,V,X_2}(u,x_1|s,v,x_2)\\
\sum_{u,v,x_1,x_2}P_{S,U,V,X_1,X_2}(s,u,v,x_1,x_2) &= Q_S(s).
\end{align}
\label{MeasureForCapacityRegionDiscreteMemorylessChannel}
\end{subequations}

The relations in \eqref{MeasureForCapacityRegionDiscreteMemorylessChannel} imply that $(U,V) \leftrightarrow (S,X_1,X_2) \leftrightarrow Y$ is a Markov chain, and $X_2$ is independent of $S$.

Define $\mc C$ to be the set of all rate pairs $(R_c,R_1)$ such that 
\begin{align}
R_1 \: &\leq \: I(U;Y|V,X_2)-I(U;S|V,X_2) \nonumber\\
R_c+ R_1 \: &\leq \: I(U,V,X_2;Y)-I(U,V,X_2;S)\nonumber\\
&\hspace{2cm} \text{for some}\:\: (S,U,V,X_1,X_2,Y) \in \mc P.
\label{CapacityRegionDiscreteMemorylessChannel}
\end{align}

The following proposition states some properties of $\mc C$.

\begin{proposition}\label{Proposition__Properties__of__CapacityRegion}

{\color{white} (properties of capacity region)}

\begin{itemize}
\item[1.] The set $\mc C$ is convex.
%\item[2.] To exhaust $\mc C$, it is enough to take $X_1$ to be a deterministic function of $(S,U,V,X_2)$.
\item[2.] To exhaust $\mc C$, it is enough to restrict $\mc V$ and $\mc U$ to satisfy
\begin{subequations}
\begin{align}
\label{BoundsOnCardinalityOfAuxiliaryRandonVariableV__CapacityRegion__DiscreteMemorylessChannel}
&|\mc V| \leq |\mc S||\mc X_1||\mc X_2|+1\\
&|\mc U| \leq \Big(|\mc S||\mc X_1||\mc X_2|+1\Big)|\mc S||\mc X_1||\mc X_2|.
\label{BoundsOnCardinalityOfAuxiliaryRandonVariableU__CapacityRegion__DiscreteMemorylessChannel}
\end{align}
\label{BoundsOnCardinalityOfAuxiliaryRandonVariables__CapacityRegion__DiscreteMemorylessChannel}
\end{subequations}
\end{itemize}
\end{proposition}

\textbf{Proof:} The proof of Proposition~\ref{Proposition__Properties__of__CapacityRegion} appears in Appendix~\ref{appendixProposition__Properties__of__CapacityRegion}. 

%-----------------------------------------------------------------------
\noindent As stated in the following theorem, the set $\mc C$ characterizes the capacity region of the state-dependent discrete memoryless MAC model that we study.

\vspace{0.3cm}

\begin{theorem}\label{Theorem__CapacityRegionDiscreteMemorylessChannel}
The capacity region of the multiple access channel with states known only strictly causally at the encoder that sends the common message and non-causally at the encoder that sends both messages is given by $\mc C$.
\end{theorem}

\vspace{0.3cm}

\textbf{Proof:} An outline proof of the coding scheme that we use for the direct part will follow. The associated error analysis and the proof of the converse appear in Appendix~\ref{appendixTheorem__CapacityRegionDiscreteMemorylessChannel}.

Theorem~\ref{Theorem__CapacityRegionDiscreteMemorylessChannel} continues to hold if in \eqref{MeasureForCapacityRegionDiscreteMemorylessChannel} we replace $P_{U|S,V,X_2}$ by $P_{U|S,V}$. Also, it should be noted that setting $V=\emptyset$ in \eqref{CapacityRegionDiscreteMemorylessChannel}, the capacity region $\mc C$ reduces to the union of all rate-pairs $(R_c,R_1)$ satisfying 
\begin{align}
R_1 \: &\leq \: I(U;Y|X_2)-I(U;S|X_2)\nonumber\\
R_c+ R_1 \: &\leq \: I(U,X_2;Y)-I(U,X_2;S)
\label{CapacityRegionDiscreteMemorylessChannel__NoStatesatEncoder2}
\end{align}
for some measure on $\mc S{\times}\mc U{\times}\mc X_1{\times}\mc X_2{\times}\mc Y$ of the form
\begin{align}
P_{S,U,X_1,X_2,Y} &= Q_SP_{X_2}P_{U,X_1|S,X_2}W_{Y|X_1,X_2,S}.
\label{MeasureCapacityRegionDiscreteMemorylessChannel__NoStatesatEncoder2}
\end{align}
Let $\mc C'$ denote the region defined by \eqref{CapacityRegionDiscreteMemorylessChannel__NoStatesatEncoder2} and \eqref{MeasureCapacityRegionDiscreteMemorylessChannel__NoStatesatEncoder2} in the remaining of this paper. It has been shown in \cite{SBSV07a} that the region $\mc C'$ is the capacity region of the MAC model of Figure~\ref{ModelForMACwithAsymmetricCSI} but with the states completely unknown at Encoder 2, i.e., while the encoding at Encoder 1 is given by \eqref{EncodingFunction__Encoder2}, the encoding at Encoder 2 is defined by the mapping 
\begin{align}
\phi_{2}: \mc W_c \longrightarrow \mc X^n_2.
\label{EncodingFunction__Encoder2__NoStatesatEncoder2}
\end{align}
Observing that $\mc C' \subseteq \mc C$ shows that the knowledge of the states only strictly causally at Encoder 2 in our model in general increases the capacity region. In Section~\ref{secIII_subsecB} we will show that the inclusion can be \textit{strict}, i.e., $\mc C' \subsetneq \mc C$.

Furthermore, one can easily check that in the case of a channel that does not depend on the states, i.e., $W_{Y|X_1,X_2,S}=W_{Y|X_1,X_2}$, the capacity region $\mc C$ reduces to the closure of the union of all rate-pairs $(R_c,R_1)$ satisfying 
\begin{align}
R_1 \: &\leq \: I(X_1;Y|Z,X_2)\nonumber\\
R_c+ R_1 \: &\leq \: I(X_1,X_2;Y)
\end{align}
for some 
\begin{align}
P_{Z,X_1,X_2,Y} &= P_{Z}P_{X_1|Z}P_{X_2|Z}W_{Y|X_1,X_2}.
\end{align}
Also, it is noted that Theorem~\ref{Theorem__CapacityRegionDiscreteMemorylessChannel} remains intact if we allow feedback to Encoder 1, i.e., before producing the $i$th channel input symbol, Encoder 1 also observes the past channel output sequence $Y^{i-1}$. That is, the encoding at Encoder 2 is still given by \eqref{EncodingFunction__Encoder2} and  that at Encoder 1 is replaced by a sequence of mappings $\{\phi_{1,i}\}^n_{i=1}$, with
\begin{align}
\phi_{1,i}: \mc W_c{\times}\mc W_1{\times}\mc S^n {\times}\mc Y^{i-1} \longrightarrow \mc X_1.
\label{EncodingFunctionWithFeedback__Encoder1}
\end{align}

We now turn to the proof of achievability of Theorem~\ref{Theorem__CapacityRegionDiscreteMemorylessChannel}. The following remark is useful for a better understanding of the coding scheme that we use to establish the achievability of Theorem \ref{Theorem__CapacityRegionDiscreteMemorylessChannel}.

\begin{remark}\label{remark1}
The proof of achievability of Theorem~\ref{Theorem__CapacityRegionDiscreteMemorylessChannel} is based on a block-Markov coding scheme in which a lossy version of the state is conveyed to the decoder, in the spirit of \cite{LS10a,LS10b,LS11a}, in addition to a generalized Gel'fand-Pinsker binning for the transmission of the information messages \cite{GP80}. However, unlike \cite{LS10a,LS10b} and \cite{LS11a} where Wyner-Ziv compression \cite{WZ76} is utilized for the transmission of the lossy version of the state, here, inspired by the noisy network coding scheme of \cite{H-LKGC11}, at each block the compression index of the state of the previous block is sent using standard rate distortion, not Wyner-Ziv binning. Also, unlike \cite{LS10a,LS10b} and \cite{LS11a} where every information message is divided into blocks and different submessages are sent over these blocks and then decoded one at a time using the same codebook as in the original compress-and-forward scheme by Cover and El Gamal \cite{CG79}, here the \textit{entire} common message and the \textit{entire} individual message are transmitted over \textit{all} blocks using codebooks that are generated independently, one for each block, and the decoding is performed simultaneously using all blocks as in \cite{H-LKGC11}. At the end of the transmission, the receiver uses the outputs of all blocks to perform simultaneous decoding of the information common and individual messages, without uniquely decoding the compression indices. \qed 
%As the analysis of our coding scheme will show, it is the combination of the messages repetition encoding, state description without binning and simultaneous nonunique decoding which permits to improve the rate region of \cite{ZPS11a}, by allowing a more general input distribution and relaxing the constraint \eqref{NonNegativityConstraint}. \qed 
\end{remark}

\textbf{Proof of Achievability:} 

The transmission takes place in $B$ blocks. The common message $W_c$ and the individual message $W_1$ are sent over \textit{all} blocks. We thus have $B_{W_c}=nB{R_c}$, $B_{W_1}=nB{R_1}$, $N=nB$, $R_{W_c}=B_{W_c}/N=R_c$ and $R_{W_1}=B_{W_1}/N=R_1$, where $B_{W_c}$ is the number of common message bits, $B_{W_1}$ is the number of individual message bits, $N$ is the number of channel uses and $R_{W_c}$ and $R_{W_1}$ are the overall rates of the common and individual messages, respectively.

\noindent \textbf{Codebook Generation:} Fix a measure $P_{S,U,V,X_1,X_2,Y} \in \mc P$. Fix $\epsilon > 0$, $\eta_c > 0$, $\eta_1 > 0$, $\hat{\eta} > 0$, $\delta > 1$ and denote $M_c = 2^{nB[R_c-\eta_c\epsilon]}$, $M_1 = 2^{nB[R_1-\eta_1\epsilon]}$, $\hat{M} = 2^{n[\hat{R}+\hat{\eta}\epsilon]}$ and $J=2^{n[I(U;S|V,X_2)+\delta\epsilon]}$.

\noindent We randomly and independently generate a codebook for each block.

\begin{itemize}
\item[1)] For each block $i$, $i=1,\hdots,B$, we generate $M_c\hat{M}$ independent and identically distributed (i.i.d.) codewords $\dv x_{2,i}(w_c,t'_i)$ indexed by $w_c=1,\hdots,R_c$, $t'_i=1,\hdots,\hat{M}$, each with i.i.d. components drawn according to $P_{X_2}$.
\item[2)] For each block $i$, for each codeword $\dv x_{2,i}(w_c,t'_i)$,  we generate $\hat{M}$ i.i.d. codewords $\dv v_i(w_c,t'_i,t_i)$ indexed by $t_i=1,\hdots,\hat{M}$, each with i.i.d. components drawn according to $P_{V|X_2}$.
\item[3)]  For each block $i$, for each codeword $\dv x_{2,i}(w_c,t'_i)$, for each codeword $\dv v_i(w_c,t'_i,t_i)$, we generate a collection of $JM_1$ i.i.d. codewords $\{\dv u_i(w_c,t'_i,t_i,w_1,j_i)\}$ indexed by $w_1=1,\hdots,M_1$, $j_i=1,\hdots,J$, each with i.i.d. components draw according to $P_{U|V,X_2}$.
\end{itemize}

\textbf{Encoding:} Suppose that a common message $W_c=w_c$ and an individual message $W_1=w_1$ are to be transmitted. As we mentioned previously, $w_c$ and $w_1$ will be sent over \textit{all} blocks. We denote by $\dv s[i]$ the state affecting the channel in block $i$, $i=1,\hdots,B$. For convenience, we let $\dv s[0]=\emptyset$ and $t_{-1}=t_0=1$ (a default value). The encoding at the beginning of block $i$, $i=1,\hdots,B$, is as follows.

\noindent Encoder $2$, which has learned the state sequence $\dv s[i-1]$, knows $t_{i-2}$ and looks for a compression index $t_{i-1} \in [1:\hat{M}]$ such that $\dv v_{i-1}(w_c,t_{i-2},t_{i-1})$ is strongly jointly typical with $\dv s[i-1]$ and $\dv x_{2,i-1}(w_c,t_{i-2})$. If there is no such index or the observed state $\dv s[i-1]$ is not typical, $t_{i-1}$ is set to $1$ and an error is declared. If there is more than one such index $t_{i-1}$, choose the smallest. Encoder 2 then transmits the vector $\dv x_{2,i}(w_c,t_{i-1})$.

\noindent Encoder 1 obtains $\dv x_{2,i}(w_c,t_{i-1})$ similarly. It then finds the smallest compression index $t_i \in [1:\hat{M}]$ such that $\dv v_i(w_c,t_{i-i},t_i)$ is strongly jointly typical with $\dv s[i]$ and $\dv x_{2,i}(w_c,t_{i-1})$. Again, if there is no such index or the observed state $\dv s[i]$ is not typical, $t_i$ is set to $1$ and an error is declared. Next, Encoder 1 looks for the smallest $j_{i}$ such that $\dv u_i(w_c,t_{i-1},t_i,w_1,j_{i})$ is jointly typical with $\dv s[i]$  given $(\dv x_{2,i}(w_c,t_{i-1}),\dv v_i(w_c,t_{i-1},t_i))$. Denote this $j_{i}$ by $j^{\star}_{i}=j(\dv s[i],w_c,t_{i-1},t_i,w_1)$. If such $j^{\star}_{i}$ is not found, an error is declared and $j(\dv s[i],w_c,t_{i-1},t_i,w_1)$ is set to $j_{i}=J$. Encoder 1 then transmits a vector $\dv x_1[i]$ which is drawn i.i.d. conditionally given $\dv u_i(w_c,t_{i-1},t_i,w_1,j^{\star}_{i})$, $\dv s[i]$, $\dv v_i(w_c,t_{i-1},t_i)$ and $\dv x_{2,i}(w_c,t_{i-1})$ (using the conditional measure $P_{X_1|U,S,V,X_2}$ induced by \eqref{MeasureForCapacityRegionDiscreteMemorylessChannel}).

\textbf{Decoding:} At the end of the transmission, the decoder has collected all the blocks of channel outputs $\dv y[1],\hdots,\dv y[B]$.

\noindent \underline{\textit{Step (a):}} The decoder estimates message $w_c$ using \text{all} blocks $i=1,\hdots,B$, i.e., simultaneous decoding. It declares that $\hat{w}_c$ is sent if there exist $t^B=(t_1,\hdots,t_B) \in [1:\hat{M}]^{B}$, $w_1 \in [1:M_1]$ and $j^B=(j_{1},\hdots,j_{B}) \in [1:J]^B$ such that $\dv x_{2,i}(\hat{w}_c,t_{i-1})$, $\dv u_i(\hat{w}_c,t_{i-1},t_i,w_1,j_{i})$, $\dv v_i(\hat{w}_c,t_{i-1},t_i)$ and $\dv y[i]$ are jointly typical for all $i=1,\hdots,B$. One can show that the decoder obtains the correct $w_c$ as long as $n$ and $B$ are large and
\begin{align}
R_c + R_1 &\leq I(U,V,X_2;Y)-I(U,V,X_2;S).
\label{Constraint__On__SumRate}
\end{align}

\noindent \underline{\textit{Step (b):}} Next, the decoder  estimates message $w_1$ using again \text{all} blocks $i=1,\hdots,B$, i.e., simultaneous decoding. It declares that $\hat{w}_1$ is sent if there exist $t^B=(t_1,\hdots,t_B) \in [1:\hat{M}]^{B}$, $j^B=(j_{1},\hdots,j_{B}) \in [1:J]^B$ such that $\dv x_{2,i}(\hat{w}_c,t_{i-1})$, $\dv u_i(\hat{w}_c,t_{i-1},t_i,\hat{w}_1,j_{i})$, $\dv v_i(\hat{w}_c,t_{i-1},t_i)$ and $\dv y[i]$ are jointly typical for all $i=1,\hdots,B$. One can show that the decoder obtains the correct $w_1$ as long as $n$ and $B$ are large and
\begin{subequations}
\begin{align}
\label{Constraint__On__IndividualRate}
R_1 &\leq I(U;Y|V,X_2)-I(U;S|V,X_2)\\
R_1 &\leq I(U,V,X_2;Y)-I(U,V,X_2;S).
\end{align}
\end{subequations}
\hspace{14cm} \qed

In the coding scheme of Theorem~\ref{Theorem__CapacityRegionDiscreteMemorylessChannel}, the state compression is standard, i.e., uses no Wyner-Ziv binning, the same message is sent in every block, and the decoding of the sent message is performed jointly using all blocks. Although of no benefit in the case of one relay, the combination of these three features was shown to be essential in achieving rates that are strictly larger than those offered by schemes based on Cover and El Gamal classic compress-and-forward scheme \cite{CG79} for certain networks with multiple relays in \cite{H-LKGC11}. That is, the coding scheme of \cite{H-LKGC11} outperforms Cover and El Gamal classic compress-and-forward for some multi-relay networks in \cite{H-LKGC11}. One can wonder whether the same holds for our model, i.e., whether schemes based on Cover and El Gamal classic compress-and-forward, i.e., block Markov encoding combined with Wyner-Ziv binning, fall short of achieving optimality for our model. In this paper, we show that the capacity region $\mc C$ as given by \eqref{CapacityRegionDiscreteMemorylessChannel} can be achieved alternatively with a coding scheme that we obtain by building upon and modifying Cover and El Gamal original compress-and-forward scheme. The modification consists essentially in 1) decoding block-by-block backwardly instead of block-by-block forwardly and 2) non-unique decoding of the compression indices. (In fact, by investigating more closely the converse proof of Theorem~\ref{Theorem__CapacityRegionDiscreteMemorylessChannel}, we will show later that 2) can be relaxed essentially without altering the capacity region). The following theorem states the result.

\vspace{0.3cm}

\begin{theorem}\label{Theorem__WynerZivBinningOptimality}
For the state-dependent multiaccess channel model that we study, there exists an optimal coding scheme that uses Wyner-Ziv binning for the state compression. That is, the capacity region $\mc C$ given by \eqref{CapacityRegionDiscreteMemorylessChannel} can also be achieved using a coding scheme in which the state compression is performed using Wyner-Ziv binning.
\end{theorem}

\vspace{0.3cm}

\textbf{Proof:} The achievability proof of Theorem~\ref{Theorem__WynerZivBinningOptimality} is based on a block-Markovian coding scheme that combines carefully Gel'fand-Pinsker binning and Wyner-Ziv binning, and utilizes backward decoding with non-unique decoding of the compression indices. The complete proof of Theorem~\ref{Theorem__WynerZivBinningOptimality} is given in Appendix~\ref{appendixTheorem__WynerZivBinningOptimality}.

As we mentioned previously, the coding scheme  of Theorem~\ref{Theorem__WynerZivBinningOptimality} shares elements with Cover and El Gamal original compress-and-forward \cite[Theorem 7]{CG79}; but differs from it mainly in two aspects. First, it uses backward decoding instead of the forward decoding of \cite{CG79}; and, second,  unlike \cite{CG79} it does not require unique decoding of the compression indices. The second aspect is essential for getting the \textit{same} rate expression as in \eqref{CapacityRegionDiscreteMemorylessChannel}, with no additional constraints. However, as we will see shortly in the corollary that will follow, one can modify the coding scheme of Theorem~\ref{Theorem__WynerZivBinningOptimality} in a way to get the compression indices decoded uniquely and \textit{still} get the capacity region, at the expense of slightly larger $|\mc V|$ and larger  $|\mc U|$. The key element is the observation that the constraint introduced by getting the compression index decoded, i.e., (see Appendix~\ref{appendixCorollary__EquivalentCharacterizationCapacityRegionDiscreteMemorylessChannel})
\begin{align}
I(V;S|X_2)-I(V;Y|X_2) &\leq I(X_2;Y),
\label{ConstraintOuterBound__Form2}
\end{align}
or, equivalently,
\begin{align}
I(V,X_2;Y)-I(V,X_2;S) &\geq 0,
\label{ConstraintOuterBound__Form1}
\end{align}
is also \textit{implicit} in the converse proof of Theorem~\ref{Theorem__CapacityRegionDiscreteMemorylessChannel}. That is, the auxiliary random variables  $U$ and $V$ of the converse proof of Theorem~\ref{Theorem__CapacityRegionDiscreteMemorylessChannel} in Appendix~\ref{appendixTheorem__CapacityRegionDiscreteMemorylessChannel} satisfy \eqref{ConstraintOuterBound__Form1}. 

\vspace{0.3cm}

\begin{corollary}\label{Corollary__EquivalentCharacterizationCapacityRegionDiscreteMemorylessChannel}
The coding scheme of Theorem~\ref{Theorem__WynerZivBinningOptimality} can be modified in a way to get the compression index decoded. The resulting coding scheme is optimal and achieves an equivalent characterization of the capacity region of the model that we study given by the set of all rate pairs $(R_c,R_1)$ such that
\begin{align}
R_1 \: &\leq \: I(U;Y|V,X_2)-I(U;S|V,X_2) \nonumber\\
R_c+ R_1 \: &\leq \: I(U,V,X_2;Y)-I(U,V,X_2;S)
\label{EquivalentCharacterizationCapacityRegionDiscreteMemorylessChannel}
\end{align}
for some measure $(S,U,V,X_1,X_2,Y) \in \mc P$ and satisfying
\begin{align}
I(V,X_2;Y)-I(V,X_2;S) &\geq 0,
\label{ConstraintOuterBound}
\end{align}
where the auxiliary random variables $V$ and $U$ have their alphabets bounded as
\begin{subequations}
\begin{align}
\label{BoundsOnCardinalityOfAuxiliaryRandonVariableV__EquivalentCharacterizationCapacityRegion__DiscreteMemorylessChannel}
&|\mc V| \leq |\mc S||\mc X_1||\mc X_2|+2\\
&|\mc U| \leq \Big(|\mc S||\mc X_1||\mc X_2|+2\Big)|\mc S||\mc X_1||\mc X_2|.
\label{BoundsOnCardinalityOfAuxiliaryRandonVariableU__EquivalentCharacterizationCapacityRegion__DiscreteMemorylessChannel}
\end{align}
\label{BoundsOnCardinalityOfAuxiliaryRandonVariables__EquivalentCharacterizationCapacityRegion__DiscreteMemorylessChannel}
\end{subequations}

\end{corollary}

\vspace{0.3cm}

\textbf{Proof:} The coding scheme that we use for the proof of Corollary~\ref{Corollary__EquivalentCharacterizationCapacityRegionDiscreteMemorylessChannel} is very similar to that of Theorem~\ref{Theorem__WynerZivBinningOptimality}, but with unique decoding of the compression indices. The details of the proof are given in Appendix~\ref{appendixCorollary__EquivalentCharacterizationCapacityRegionDiscreteMemorylessChannel}.

\vspace{0.3cm}

We now establish an alternative outer bound on the capacity region of the DM MAC model that we study. This outer bound will turn out to be useful in the proof of the converse part of the coding theorem for the Gaussian case in Section~\ref{secIV} since, as it will be shown, it is also achievable in that case.

\begin{theorem}\label{Theorem__AlternativeOuterBoundDiscreteMemorylessChannel}

The capacity region of the multiple access channel with states known non-causally at one encoder and strictly causally at the other encoder is contained in the closure of the set of all rate-pairs $(R_c,R_1)$ satisfying
\begin{align}
R_1 \: &\leq \: I(X_1;Y|S,X_2) \nonumber\\
R_c+ R_1 \: &\leq \: I(X_1,X_2;Y|S)-I(X_2;S|Y),
\label{AlternativeOuterBoundDiscreteMemorylessChannel}
\end{align}
for some probability distribution of the form
\begin{align}
&P_{S,X_1,X_2,Y}=Q_SP_{X_2}P_{X_1|X_2,S}W_{Y|X_1,X_2,S}.
\label{MeasureForAlternativeOuterBoundDiscreteMemorylessChannel}
\end{align}
\end{theorem}

\textbf{Proof:} The proof of Theorem~\ref{Theorem__AlternativeOuterBoundDiscreteMemorylessChannel} appears in Appendix~\ref{appendixTheorem__AlternativeOuterBoundDiscreteMemorylessChannel}.

\begin{remark}\label{remark4}
In \cite{SBSV07a} the authors use an extension of the converse part of the proof of the standard Gel'fand-Pinsker capacity to establish a converse proof for the model with states $S^n$ known non-causally at Encoder 1 and no states at all at Encoder 2. Then, they show that their outer bound, which involves an auxiliary random variable, is itself contained in the region defined by \eqref{AlternativeOuterBoundDiscreteMemorylessChannel}. In Appendix~\ref{appendixTheorem__AlternativeOuterBoundDiscreteMemorylessChannel}, we provide a direct proof that the region defined by \eqref{AlternativeOuterBoundDiscreteMemorylessChannel} is an outer bound on the capacity region of the more general model that we study here. Our converse proof accounts also for the availability of the states at Encoder 2 in a strictly causal manner. \qed
\end{remark}

\subsection{Example}\label{secIII_subsecB}

In Section~\ref{secIII_subsecA} we have shown that the capacity region $\mc C$ of the model of Figure~\ref{ModelForMACwithAsymmetricCSI} is potentially larger than that, $\mc C'$, of the same model but with Encoder 2 being totally unaware of the states, i.e., $\mc C' \subseteq \mc C$. In this section, we show that this inclusion can be \textit{strict}, i.e., $\mc C' \subsetneq \mc C$.

We use $h(\alpha)$ to denote the entropy of a Bernoulli\:$(\alpha)$ source, i.e.,
\begin{equation}
h(\alpha) = - \alpha \log(\alpha) - (1-\alpha)\log(1-\alpha)
\end{equation}
and $p * q$ to denote the binary convolution, i.e.,
\begin{equation}
p * q = p(1-q)+q(1-p).
\end{equation}

Consider the binary memoryless MAC shown in Figure~\ref{BSCModelCounterExample}. Here, all the random variables are binary $\{0,1\}$. The channel has two output components, i.e., $Y^n=(Y^n_1,Y^n_2)$. The component $Y^n_2$ is deterministic, $Y^n_2=X^n_2$, and the component $Y^n_1=X^n_1 + S^n + Z^n_1$, where the addition is modulo $2$. Encoder 2 knows the states only strictly causally and has no message to transmit. Encoder 1 knows the states non-causally and transmits an individual message $W_1$. The state and noise vectors are independent and memoryless, with the state process $S_i$, $i \geq 1$, and the noise process $Z_{1,i}$, $i \geq 1$, assumed to be Bernoulli $(\frac{1}{2})$ and  Bernoulli $(p)$ processes, respectively. The vectors $X^n_1$ and $X^n_2$ are the channel inputs, subjected to the constraints
\begin{align}
\sum_{i=1}^{n} X_{1,i} &\leq nq_1 \quad \text{and} \quad \sum_{i=1}^{n} X_{2,i} \leq nq_2.
\label{BinaryChannel__InputsConstraints}
\end{align}

\begin{figure}[htpb]
\centering
\includegraphics[width=0.5\linewidth]{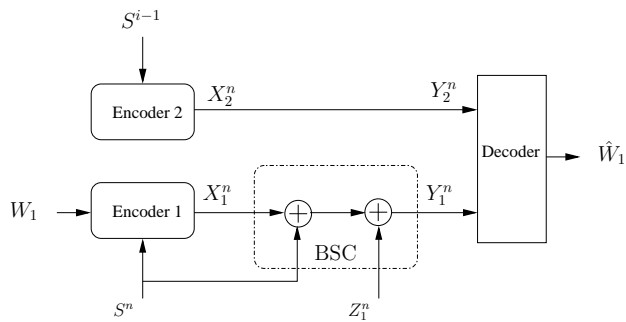}
\caption{Binary state-dependent MAC example with two output components, $Y^n=(Y^n_1,Y^n_2)$, with $Y^n_1=X^n_1 + S^n + Z^n_1$ and $Y^n_2=X^n_2$.}
\label{BSCModelCounterExample}
\end{figure}

For this example, as we will show shortly, the strictly causal knowledge of the states at Encoder 2 \textbf{does} help, and in fact Encoder 1 can transmit at rates that are larger than the standard Gel'fand-Pinsker $I(U;Y_1)-I(U;S)$ which would be the capacity had Encoder 2 been of no help.

\textit{Claim 1:} The capacity of the state-dependent binary memoryless MAC shown in Figure~\ref{BSCModelCounterExample} is given by
\begin{align}
C_B &= \max_{p(x_1|s)} \:\: I(X_1;Y_1|S).
% &= h(p * q_1) - h(p).
\label{Capacity__ModelCounterExample}
\end{align}

\textit{Proof:} 1) The achievability follows from Theorem~\ref{Theorem__CapacityRegionDiscreteMemorylessChannel}, as follows. Set $R_c=0$ and $V=S$, $U=X_1$, $Y_2=X_2$ with $X_2$ independent of $(S,X_1)$ in Theorem~\ref{Theorem__CapacityRegionDiscreteMemorylessChannel}. Evaluating the first inequality, we obtain
\begin{align}
R_1 &\leq I(U;Y|V,X_2)-I(U;S|V,X_2)\\
    &= I(X_1;Y_1,X_2|S,X_2)\\
    &= I(X_1;Y_1|S,X_2)\\
    &= I(X_1,X_2;Y_1|S)-I(X_2;Y_1|S)\\
    &= I(X_1;Y_1|S)+I(X_2;Y_1|X_1,S)-I(X_2;Y_1|S)\\
\label{Constraint1__MaximalIndividualRate__ModelCounterExample__Step1}
    &= I(X_1;Y_1|S)-I(X_2;Y_1|S)\\
    &= I(X_1;Y_1|S),
 \label{Constraint1__MaximalIndividualRate__ModelCounterExample}
\end{align}
where \eqref{Constraint1__MaximalIndividualRate__ModelCounterExample__Step1} follows since $X_2=Y_2$ and $Y_2 \leftrightarrow (X_1,S) \leftrightarrow Y_1$ is a Markov chain, and  the last equality follows by the Markov relation $X_2 \leftrightarrow  S \leftrightarrow Y_1$ for this example.

Evaluating the second inequality, we obtain
\begin{align}
R_1 &\leq I(U,V,X_2;Y)-I(U,V,X_2;S)\\
    &= I(X_1,S;Y_1,X_2)+H(X_2|X_1,S)-H(S)\\
    &= I(X_1,S;Y_1)+I(X_1,S;X_2|Y_1)+H(X_2|X_1,S)-H(S)\\
    &= I(X_1,S;Y_1)+H(X_2|Y_1)-H(X_2|X_1,S,Y_1)+H(X_2|X_1,S)-H(S)\\
\label{Constraint2__MaximalIndividualRate__ModelCounterExample__Step1}
    &= I(X_1;Y_1|S)+I(S;Y)+H(X_2|Y_1)-H(S)\\
    &= I(X_1;Y_1|S)+H(X_2|Y_1)-H(S|Y_1)\\
    &= I(X_1;Y_1|S)+H(Y_1|X_2)-H(Y_1|S)+H(X_2)-H(S)\\
    &= I(X_1;Y_1|S)+H(S;Y_1)+H(X_2)-H(S)
\label{Constraint2__MaximalIndividualRate__ModelCounterExample}
\end{align}
where \eqref{Constraint2__MaximalIndividualRate__ModelCounterExample__Step1} follows since $X_2$ is independent of $(X_1,S,Y_1)$.

Now, observe that with the choice $X_2 \sim \: \text{Bernoulli}\:(\frac{1}{2})$ independent of $(S,X_1)$, we have $H(X_2)=H(S)=1$ and, so, the RHS of \eqref{Constraint2__MaximalIndividualRate__ModelCounterExample} is larger than the RHS of \eqref{Constraint1__MaximalIndividualRate__ModelCounterExample}. This shows the achievability of the rate $R_1 = I(X_1;Y_1|S)$.

2) The converse follows straightforwardly by specializing Theorem 2 (or the cut-set upper bound) to this example, 
\begin{align}
R  &\leq I(X_1;Y|X_2,S)\\
   &= I(X_1;Y_1|X_2,S)\\
   &= H(Y_1|X_2,S)-H(Y_1|X_1,X_2,S)\\
\label{ProofCutSetBoundCounterExample__Step1}
   &\leq H(Y_1|S)-H(Y_1|X_1,X_2,S)\\
   &\leq  H(Y_1|S)-H(Y_1|X_1,S)\\
\label{ProofCutSetBoundCounterExample__Step2}
   &= I(X_1;Y_1|S),
\end{align}
where \eqref{ProofCutSetBoundCounterExample__Step1} holds since conditioning reduces entropy, and \eqref{ProofCutSetBoundCounterExample__Step2} holds by the Markov relation $X_2 \leftrightarrow (X_1,S) \leftrightarrow Y_1$.

\textit{Claim 2:} The capacity of the state-dependent binary memoryless MAC shown in Figure~\ref{BSCModelCounterExample} satisfies
\begin{align}
C_B &= h(p * q_1) - h(p) > \max_{p(u,x_1|s)} \:\: I(U;Y_1)-I(U;S).
\label{ExplicitCharacterization__Capacity__ModelCounterExample}
\end{align}

\textit{Proof:} Claim 2 is a simple consequence of Claim 1 and known results on the capacity of the binary dirty paper channel (see for example \cite{PCR03} and references therein). More specifically, the capacity $C_B$ in Claim 1 is that of a point-to-point state-dependent additive binary channel with a Bernoulli\:$(\frac{1}{2})$ state known at both transmitter and receiver ends, a Bernoulli\:$(p)$ noise representing the binary symmetric channel and average input constraint $q_1$ at the transmitter. Thus, an explicit characterization of $C_B$ is given by \cite{PCR03}
\begin{align}
C_B &= h(p * q_1) - h(p).
\end{align}

Let now $R_{\text{GP}}$ be the maximum achievable rate had the strictly causal part $S^{i-1}$ of the state been of no utility, or equivalently, had Encoder 2 been of no help. $R_{\text{GP}}$ is the capacity of a binary dirty paper channel given by \cite{PCR03}
\begin{align}
R_{\text{GP}} &= \max_{p(u,x_1|s)} \:\: I(U;Y_1)-I(U;S)\nonumber\\
&= \left\{
\begin{array}{ll}
G(q_1) &\text{if}\:\:\: p^{\star} \leq q_1 \leq \frac{1}{2}\\
q_1 \log(\frac{1-p^{\star}}{p^{\star}}) &\text{if}\:\:\: 0 \leq q_1 \leq p^{\star}
\end{array} \right\}
\label{Gelfand-PinskerRate__ModelCounterExample}
\end{align}
where $p^{\star}=1-2^{-h(p)}$ and the function $G(q)$, defined for $q \in [0,1/2]$, is given by
\begin{align}
G(q) = 
\left\{
\begin{array}{ll}
h(q) - h(p) &\text{if}\:\:\: p \leq q \leq \frac{1}{2}\\
0 &\text{if}\:\:\: 0 \leq q \leq p
\end{array}
\right.
\end{align}
Observing that $h(p * q_1) > h(q_1)$ for all $0 < q_1 < 1/2$, it is easy to see that $C_B > R_{\text{GP}}$.

\begin{remark}
In this example, the encoder that knows the states only strictly causally simply conveys these states to the receiver, noiselessly. The receiver then becomes aware of the channel states fully (since the delay in learning these states at the decoder has no impact on the capacity). This explains why Encoder 1 can transmit at rates that can be strictly larger than the standard Gel'fand-Pinker rate \eqref{Gelfand-PinskerRate__ModelCounterExample}; and in fact achieves the capacity \eqref{ExplicitCharacterization__Capacity__ModelCounterExample} of a state-dependent additive binary channel with the states known at both transmitter and receiver ends. \qed
\end{remark}

\subsection{Common-message Capacity}\label{secIII_subsecC}

In this section, we study the important case in which the two encoders transmit only the common message, i.e., $R_1=0$. The following corollary characterizes the capacity in this case, to which we refer as \textit{common-message capacity}.

\begin{corollary}\label{Corollary__CommonMessageCapacity}
The common message capacity, $C$, of the multiple access channel with common message and states known non-causally at one encoder and strictly causally at the other encoder is given by
\begin{align}
C &= \max I(K,X_2;Y)-I(K,X_2;S)
\end{align}
where the maximization is over joint measures $P_{S,K,X_1,X_2,Y}$  of the form
\begin{align}
P_{S,K,X_1,X_2,Y} &= Q_SP_{X_2}P_{K,X_1|S,X_2}.
\end{align}
\end{corollary}

\textbf{Proof:} The proof of Corollary~\ref{Corollary__CommonMessageCapacity} appears in Appendix~\ref{appendixCorollary__CommonMessageCapacity}.

\begin{remark}\label{remark5}
The common-message capacity of our model in Corollary~\ref{Corollary__CommonMessageCapacity} coincides with the common-message of the model with the state sequence $S^n$ known noncausally at Encoder 1 and not at all at Encoder 2 \cite{SBSV07a}. That is, $C$ can also be obtained by relaxing the constraint on $R_1$ in the region $\mc C'$ defined by \eqref{CapacityRegionDiscreteMemorylessChannel__NoStatesatEncoder2} and \eqref{MeasureCapacityRegionDiscreteMemorylessChannel__NoStatesatEncoder2}. This shows that the knowledge of the states at Encoder 2 only strictly causally does not increase the common-message capacity. We should, however, note that this result is not a direct consequence of that in a MAC a state that is known only strictly causally at \textit{all} encoders does not increase the capacity; and, so, the converse proof is needed here. \qed 

% Moreover, in contrast to \cite{SBSV07a}, the converse proof does not follow directly from the converse part proof of the capacity formula for the standard Gelf'and-Pinsker channel \cite{GP80} because, at time $i$, Encoder 2 sends inputs which are function of not only the message to transmit, but also the past state sequence $S^{i-1}$. Our converse proof includes a redefinition of the auxiliary random variable.
\end{remark}

\section{Memoryless Gaussian Case}\label{secIV}

In this section, we consider a two-user state-dependent Gaussian MAC in which the channel states and the noise are additive and Gaussian.  

\subsection{Channel Model}\label{secIV_subsecA}
As in Section \ref{secII}, we assume that Encoder 1 knows the channel states non-causally and Encoder 2 knows the channel states strictly causally. The two encoders send some common message $W_c$; and, in addition, Encoder 1 sends an individual message $W_1$. At time instant $i$, the channel output $Y_i$ is related to channel inputs $X_{1,i}$ and $X_{2,i}$ from the two encoders, the channel state $S_i$ and the noise $Z_i$ by
\begin{align}
& Y_i=X_{1,i}+X_{2,i}+S_i+Z_i,
\label{ChannelModelForGaussianMACWithAsymmetricCSI}
\end{align}
where $S_i$ and $Z_i$ are zero-mean Gaussian random variables with variance $Q$ and $N$, respectively. The random variables $S_i$ and $Z_i$ at time instant $i \in \{1,\cdots,n\}$ are mutually independent, and independent from $(S_j,Z_j)$ for $j \neq i$. Also, at time $i$, the input $X_{2,i}$ is independent from the state $S_i$.

We consider the individual power constraints on the transmitted power
\begin{equation}
\sum_{i=1}^{n}X_{1,i}^2 \leq nP_1, \:\: \sum_{i=1}^{n}X_{2,i}^2 \leq nP_2.
\label{IndividualPowerConstraintsFullDuplexRegime}
\end{equation}
The definition of a code for this channel is the same as given in Section \ref{secII}, with the additional power constraints \eqref{IndividualPowerConstraintsFullDuplexRegime}.

\subsection{Capacity Region}\label{secIV_subsecB}

\noindent The following theorem characterizes the capacity region of the studied Gaussian model.

\begin{theorem}\label{Theorem__CapacityRegionMemorylessGaussianChannel}
 The capacity region of the Gaussian model \eqref{ChannelModelForGaussianMACWithAsymmetricCSI} is given by the set of all the rate pairs $(R_c,R_1)$ satisfying
 \begin{align}
 R_1\: &\leq \:\frac{1}{2}\log\Big(1+\frac{P_1(1-\rho^2_{12}-\rho^2_{2s})}{N}\Big)\nonumber\\
 R_c+R_1 \: & \leq \frac{1}{2}\log\Big(1+\frac{(\sqrt{P_2}+\rho_{12}\sqrt{P_1})^2}{P_1(1-\rho^2_{12}-\rho^2_{1s})+(\sqrt{Q}+\rho_{1s}\sqrt{P_1})^2+N}\Big)\nonumber\\&+\frac{1}{2}\log\Big(1+\frac{P_1(1-\rho^2_{12}-\rho^2_{1s})}{N}\Big),
 \label{OuterBoundGaussianChannel}
 \end{align}
 where the maximization is over $\rho_{12} \in [0,1]$, $\rho_{1s} \in [-1,0]$ such that
 \begin{equation}
 \rho^2_{12}+\rho^2_{1s} \leq 1.
 \label{AllowableCovarianceMatrixOuterBound}
 \end{equation}
 \end{theorem}

 \textbf{Proof:} An outline proof of Theorem~\ref{Theorem__CapacityRegionMemorylessGaussianChannel} is given in Appendix~\ref{appendixTheorem__CapacityRegionMemorylessGaussianChannel}.

 \begin{remark}\label{remark3}
 The capacity region of our model in Theorem~\ref{Theorem__CapacityRegionMemorylessGaussianChannel} coincides with that of the model \eqref{ChannelModelForGaussianMACWithAsymmetricCSI} but with the state sequence $S^n$ known noncausally at Encoder 1 and not all at Encoder 2 \cite[Theorem 7]{SBSV07a}. Then, an implication of Theorem~\ref{Theorem__CapacityRegionMemorylessGaussianChannel} is that it is \textit{optimal} for our model to just ignore the states $S^{i-1}$ that are known at Encoder 2 and use the coding scheme of \cite{SBSV07a}. That is, the availability of the states only strictly causally at the encoder that sends only the common message in our model does not increase the capacity region any further. While one could expect some utility of the collaborative transmission of a lossy version of the state to the decoder as in the memoryless discrete setup (and also in the Gaussian setups of \cite{LS10a,LS10b} and \cite{LSY11}), a direct consequence of our converse proof is that this would be of no help, in the sense that it would not result in better transmission rates. This can be interpreted as follows. As it can be seen from the proof of Theorem~\ref{Theorem__CapacityRegionDiscreteMemorylessChannel}, the joint transmission of the state to the decoder aims at equipping it with an estimate of this state. This state estimate is then utilized as decoder side information for the decoding of the information messages. In the discrete memoryless case, this can be beneficial in general for the transmission of the private message, not the common message, as we already mentioned. In the Gaussian case, however, for the transmission of the private message, Encoder 1 knows the state non-causally and, therefore, it can cancel its effect completely using a variation of the standard dirty paper scheme \cite{C83}, with no need to diminishing its effect via the joint transmission of the compressed version of the state. \qed

 %\noindent In the stated-dependent MAC model with strictly causal side information at the encoders studied in \cite{LS10a}, the utility of the strictly causal part of the state known at both encoders is created by utilizing a block Markov coding scheme in which, in block $i$, the encoders cooperate to send a compressed version of the state $S_{i-1}$ to the decoder in addition to their individual messages. The encoders cannot cooperate in the transmission of the messages but they do in that of the compressed version of $S_{i-1}$. The decoder estimates $S_{i-1}$ from the output received in block $i$ and then uses it to decode the messages transmitted in block $i-1$. Although some non-zero rate -- that otherwise could be used for sending the individual messages, is spent on sending the joint description of the state $S_{i-1}$, the net effect can be an increase in the capacity because, in block $i$, each encoder can benefit from the decoder's knowledge of some estimation of the state $S_{i-1}$.

 %\noindent The effect of the joint transmission of the state $S_{i-1}$ in block $i$ in Lapidoth and Steinberg coding scheme is some reduction of the state effect in block $i-1$ in decoding the individual message. Our converse proof shows that the net effect is not beneficial in our case. Intuitively, this holds because Encoder 1 knows the state non-causally and can cancel its effect completely using a variation of the standard dirty paper scheme \cite{C83}, i.e., with no need to diminishing its effect via the joint transmission of the compressed version of $S_{i-1}$.
 \end{remark}

 \noindent The following corollary follows straightforwardly from Theorem~\ref{Theorem__CapacityRegionMemorylessGaussianChannel}.

 \begin{corollary}\label{corollary1}
  The common message capacity, $C_{\text{G}}$, of the Gaussian model \eqref{ChannelModelForGaussianMACWithAsymmetricCSI} is given by
  \begin{align}
  C_{\text{G}}\: &= \max\: \frac{1}{2}\log\Big(1+\frac{(\sqrt{P_2}+\rho_{12}\sqrt{P_1})^2}{P_1(1-\rho^2_{12}-\rho^2_{1s})+(\sqrt{Q}+\rho_{1s}\sqrt{P_1})^2+N}\Big)+\frac{1}{2}\log\Big(1+\frac{P_1(1-\rho^2_{12}-\rho^2_{1s})}{N}\Big),
  \label{OuterBoundGaussianChannel}
  \end{align}
  where the maximization is over $\rho_{12} \in [0,1]$, $\rho_{1s} \in [-1,0]$ such that
  \begin{equation}
  \rho^2_{12}+\rho^2_{1s} \leq 1.
  \label{AllowableCovarianceMatrixCommonMessageCapacity}
  \end{equation}
  \end{corollary}

\section{Conclusions}\label{secV}

In this paper, we consider a state-dependent multiaccess channel with the channel state available noncausally at one of the encoders and only strictly causally at the other encoder. The decoder is not aware of the channel state. Both encoders transmit a common message and, in addition, Encoder 1 --- the encoder that knows the state noncausally, transmits an individual message. We study the capacity region of this communication model. The analysis also helps understanding the utility of revealing the state only strictly causally to the encoder that sends only the common message as well as optimal compressions to perform it.

In the discrete memoryless case, we characterize the capacity region of this model with a single-letter expression. In particular, the analysis reveals optimal ways of exploiting the knowledge of the state only strictly causally at the encoder that sends only the common message.  The encoders collaborate to convey to the decoder a lossy version of the state, in addition to transmitting the information messages through a generalized Gelfand-Pinsker binning. Particularly important in this problem are the questions of 1) optimal ways of performing the state compression, and 2) whether or not the compression indices should be decoded uniquely. We develop two optimal coding schemes that perform the state compression differently.  The first coding scheme is \`a-la noisy network coding, i.e., with no binning and non-unique decoding of the compression indices. The second coding scheme employs Wyner-Ziv binning with backward decoding and non-unique decoding of the compression indices. We note that backward decoding and non-unique decoding seem to be key elements for the optimality of the Wyner-Ziv based coding scheme. Also, we point out that the combination of these two features is likely to be beneficial in other scenarios in the context of networks with Wyner-Ziv compressions. Next, by exploiting our outer bound and the involved auxiliary variables specifically, we show that, although not required in general, for our specific model the compression indices can in fact be decoded uniquely essentially without altering the capacity region but at the expense of larger alphabets sizes for the auxiliary random variables. 

The capacity region contains that of the model of \cite{SBSV07a}, and this shows that revealing the state even only strictly causally to the encoder that sends only the common message is beneficial and enlarges the capacity region in general. Furthermore, by investigating a discrete memoryless example, we show that this inclusion can be strict, thus demonstrating the utility of conveying a compressed version of the state to the decoder cooperatively by the encoders.

We also specialize our results to the case in which the two encoders send only the common message. We characterize the common-message capacity and show that knowing the states only strictly causally at one of the encoders is not beneficial in this case.

Furthermore, we also study the memoryless Gaussian setting in which the channel state and the noise are additive and Gaussian. In this case, we establish an operative outer bound on the achievable rate pairs and then show that this outer bound is achievable; thus yielding a closed-form expression of the capacity region. Unlike the discrete memoryless case, we show that the knowledge of the states only strictly causally at the encoder that sends only the common message does not increase the capacity region in this case.

%-----------------------------------------------------------------------------------------------

\appendix
Throughout this section we denote the set of strongly jointly $\epsilon$-typical sequences \cite[Chapter 14.2]{CT91} with respect to the distribution $P_{X,Y}$ as $\mc T_{\epsilon}^n(P_{X,Y})$.

%%%%%%%%%%%%%%%%%%%%%%%%%%%%%%%%%%%%%%%%%%%%%%%%%%%%%%%%%%%%%%%%%%%%%%%%%%%%%%%%%%%%%%%%
\renewcommand{\theequation}{A-\arabic{equation}}
\setcounter{equation}{0}  % reset counter
\subsection{Proof of Proposition~\ref{Proposition__Properties__of__CapacityRegion}}\label{appendixProposition__Properties__of__CapacityRegion}

\textit{Part 1:} To prove the convexity of the region, we use a standard argument. We introduce a time-sharing random variable $T$ and define the joint distribution
\begin{align}
P_{T,S,U,V,X_1,X_2,Y}(t,s,u,v,x_1,x_2,y) &= P_{T,S,U,V,X_1,X_2}(t,s,u,v,x_1,x_2)W_{Y|X_1,X_2,S}(y|x_1,x_2,s)\\
\sum_{u,v,x_1,x_2} P_{T,S,U,V,X_1,X_2}(t,s,u,v,x_1,x_2) &= P_T(t)Q_S(s). 
\end{align}

Let now $(R^{T}_c,R^{T}_1)$ be the common and individual rates resulting from time sharing. Then,
\begin{align}
R^{T}_1 \: &\leq \: I(U;Y|V,X_2,T)-I(U;S|V,X_2,T) \\
           &= \: I(U;Y|\tilde{V},X_2)-I(U;S|\tilde{V},X_2)\\
R^{T}_c+ R^{T}_1 \: &\leq \: I(U,V,X_2;Y|T)-I(U,V,X_2;S|T)\\
                    &= \: I(U,V,X_2;Y|T)-I(U,V,X_2,T;S)\\
                    &\leq \: I(U,V,X_2,T;Y)-I(U,V,X_2,T;S)\\
                    &= \: I(U,\tilde{V},X_2;Y)-I(U,\tilde{V},X_2;S),
\end{align}
where $\tilde{V} := (V,T)$. That is, the time sharing random variable $T$ is incorporated into the auxiliary random variable $V$. This shows that time sharing cannot yield rate pairs that are not included in $\mc C$ and, hence, $\mc C$ is convex.

\textit{Part 2:} To prove that the region $\mc C$ is not altered if one restricts the random variables $U$ and $V$ to have their alphabets restricted as indicated in \eqref{BoundsOnCardinalityOfAuxiliaryRandonVariables__CapacityRegion__DiscreteMemorylessChannel}, we invoke the support lemma \cite[p. 310]{CK81}. Fix a distribution $\mu \in \mc P$ of $(S,U,V,X_1,X_2,Y)$ and, without loss of generality, let us denote the product set $\mc S\times\mc X_1\times\mc X_2=\{1,\hdots,m\}$, $m=|\mc S{\times}\mc X_1{\times}\mc X_2|$. 

\noindent To prove the bound \eqref{BoundsOnCardinalityOfAuxiliaryRandonVariableV__CapacityRegion__DiscreteMemorylessChannel} on $|\mc V|$, note that we have
\begin{align}
& I_{\mu}(U;Y|V,X_2)-I_{\mu}(U;S|V,X_2) \nonumber\\
& \hspace{1cm}= I_{\mu}(U,X_2;Y|V)-I_{\mu}(X_2;Y|V)-I_{\mu}(U,X_2;S|V)+I_{\mu}(X_2;S|V)\nonumber\\
& \hspace{1cm}= H_{\mu}(U,X_2,S|V)-H_{\mu}(U,X_2,Y|V)-H_{\mu}(X_2,S|V)+H_{\mu}(X_2,Y|V)
\end{align}
and
\begin{align}
& I_{\mu}(U,V,X_2;Y)-I_{\mu}(U,V,X_2;S) \nonumber\\
&\hspace{1cm}= I_{\mu}(U,X_2;Y|V)-I_{\mu}(U,X_2;S|V)+I_{\mu}(V;Y)-I_{\mu}(V;S)\nonumber \\
&\hspace{1cm}= H_{\mu}(U,X_2,S|V)-H_{\mu}(U,X_2,Y|V)+H_{\mu}(Y)-H_{\mu}(S).
\end{align}
\noindent Hence, it suffices to show that the following functionals of $\mu(S,U,V,X_1,X_2,Y)$
\begin{subequations}
\begin{align}
\label{Functional1__BoundsOnCardinalityOfAuxiliaryRandonVariableV__UpperBound__DiscreteMemorylessChannel}
r_{i}(\mu) &= \mu(s,x,x'), \quad i=1,\hdots,m-1\\
r_m(\mu) &= \int_{v}d_{\mu}(v)[H_{\mu}(U,X_2,S|v)-H_{\mu}(U,X_2,Y|v)-H_{\mu}(X_2,S|v)+H_{\mu}(X_2,Y|v)]\\
r_{m+1}(\mu) &= \int_{v}d_{\mu}(v)[H_{\mu}(U,X_2,S|v)-H_{\mu}(U,X_2,Y|v)]
\end{align}
\label{Functionals__BoundsOnCardinalityOfAuxiliaryRandonVariableV__UpperBound__DiscreteMemorylessChannel}
\end{subequations}
can be preserved with another measure $\mu' \in \mc P$. Observing that  there is a total of $\Big(|\mc S||\mc X_1||\mc X_2|+1\Big)$ functionals in \eqref{Functionals__BoundsOnCardinalityOfAuxiliaryRandonVariableV__UpperBound__DiscreteMemorylessChannel}, this is ensured by a standard application of the support lemma; and this shows that the alphabet of the auxiliary random variable $V$ can be restricted as indicated in \eqref{BoundsOnCardinalityOfAuxiliaryRandonVariableV__CapacityRegion__DiscreteMemorylessChannel} without altering the region $\mc C$.

\noindent Once the alphabet of $V$ is fixed, we apply similar arguments to bound the alphabet of $U$, where this time $(|\mc S||\mc X_1||\mc X_2|+1)|\mc S||\mc X_1||\mc X_2|-1$ functionals must be satisfied in order to preserve the joint distribution of $(S, V, X_1, X_2)$, and one more functional to preserve
\begin{subequations}
\begin{align}
& I_{\mu}(U;Y|V,X_2)-I_{\mu}(U;S|V,X_2) = H_{\mu}(Y,V,X_2)-H_{\mu}(S,V,X_2)+H_{\mu}(S,V,X_2|U)-H_{\mu}(Y,V,X_2|U)\\
& I_{\mu}(U,V,X_2;Y)-I_{\mu}(U,V,X_2;S) = H_{\mu}(Y)-H_{\mu}(S)+H_{\mu}(S,V,X_2|U)-H_{\mu}(Y,V,X_2|U).
\end{align}
\end{subequations}
This shows that the alphabet of the auxiliary random variable $U$ can be restricted as indicated in \eqref{BoundsOnCardinalityOfAuxiliaryRandonVariableU__CapacityRegion__DiscreteMemorylessChannel} without altering the region $\mc C$; and completes the proof of Proposition~\ref{Proposition__Properties__of__CapacityRegion}.

%%%%%%%%%%%%%%%%%%%%%%%%%%%%%%%%%%%%%%%%%%%%%%%%%%%%%%%%%%%%%%%%%%%%%%%%%%%%%%%%%%%%%%%%
\renewcommand{\theequation}{B-\arabic{equation}}
\setcounter{equation}{0}  % reset counter
\subsection{Proof of Theorem~\ref{Theorem__CapacityRegionDiscreteMemorylessChannel}}\label{appendixTheorem__CapacityRegionDiscreteMemorylessChannel}

%==============================================================================
\subsubsection{Direct Part of Theorem~\ref{Theorem__CapacityRegionDiscreteMemorylessChannel}}
To bound the probability of error, we assume without loss of generality that the compression indices are all equal to unity, i.e., $t_1=t_2=\hdots=t_B=1$.

We examine the probability of error associated with each of the encoding and decoding procedures. The events $E_1$, $E_2$ and $E_3$ correspond to encoding errors, and the events $E_4$, $E_5$, $E_6$ and $E_7$ correspond to decoding errors.

\begin{itemize}
%----------------------
\item Let $E_1=\cup_{i=1}^{B}E_{1i}$ where $E_{1i}$ is the event that, for the encoding in block $i$, there is no covering codeword $\dv v_{i-1}(w_c,t_{i-2},t_{i-1})$ strongly jointly typical with $\dv s[i-1]$ given $\dv x_{2,i-1}(w_c,t_{i-2})$, i.e.,
\begin{align}
E_1 &= \bigcup_{i=1}^{B} \Big\{\nexists \:\:t_{i-1} \in [1:\hat{M}] \:\: \text{s.t.:}\:\: \Big(\dv v_{i-1}(w_c,t_{i-2},t_{i-1}), \dv s[i-1], \dv x_{2,i-1}(w_c,t_{i-2})\Big) \in \mc T_{\epsilon}^n(P_{V,S,X_2})\Big\}.
\end{align}

\noindent For $i \in [1:B]$, the probability that $(\dv s[i-1], \dv x_{2,i-1}(w_c,t_{i-2}))$ is not jointly typical goes to zero as $n \rightarrow \infty$, by the asymptotic equipartition property (AEP) \cite[p. 384]{CT91}. Then, for $(\dv s[i-1], \dv x_{2,i-1}(w_c,t_{i-2}))$ jointly typical, the covering lemma \cite[Lecture Note 3]{GK10} ensures that the probability that there is no $t_{i-1} \in [1:\hat{M}]$ such that $(\dv v_{i-1}(w_c,t_{i-2},t_{i-1}),\dv s[i-1])$ is strongly jointly typical given $\dv x_{2,i-1}(w_c,t_{i-2})$ is exponentially small for large $n$ provided that the number of covering codewords $\dv v_{i-1}$ is greater than $2^{nI(V;S|X_2)}$, i.e.,
\begin{align}
\hat{R} &> I(V;S|X_2).
\label{Condition__on__CoveringRate}
\end{align}
Thus, if  \eqref{Condition__on__CoveringRate} holds, $\text{Pr}(E_{1i}) \rightarrow 0 \quad \text{as}\quad  n \rightarrow \infty$ and, so, by the union of bound over the $B$ blocks, $\text{Pr}(E_{1}) \rightarrow 0 \quad \text{as}\quad  n \rightarrow \infty$.

%----------------------
 \item Let $E_2=\cup_{i=1}^{B}E_{2i}$ where $E_{2i}$ is the event that, for the encoding in block $i$, Encoder 1 can find no covering codeword $\dv v_i(w_c,t_{i-1},t_i)$ strongly jointly typical with $\dv s[i]$ given $\dv x_{2,i}(w_c,t_{i-1})$. Similarly to the event $E_1$, it is easy to see that $\text{Pr}(E_2|E^c_1) \rightarrow 0 \quad \text{as}\quad  n \rightarrow \infty$ if \eqref{Condition__on__CoveringRate} is true.

%----------------------
\item Let $E_3=\cup_{i=1}^{B}E_{3i}$ where $E_{3i}$ is the event that, for the encoding in block $i$, there is no sequence $\dv u_i(w_c,t_{i-1},t_i,w_1,j_{i})$ jointly typical with $\dv s[i]$ given $\dv x_{2,i}(w_c,t_{i-1})$ and $\dv v_i(w_c,t_{i-1},t_i)$, i.e.,
\begin{align}
E_3 &= \bigcup_{i=1}^{B} \Big\{\nexists \:\:j_{i} \in [1:J] \:\: \text{s.t:}\:\: \Big(\dv u_i(w_c,t_{i-1},t_i,w_1,j_{i}), \dv s[i], \dv v_i(w_c,t_{i-1},t_i), \dv x_{2,i}(w_c,t_{i-1})\Big) \in \mc T_{\epsilon}^n(P_{U,S,V,X_2})\Big\}.
\end{align}
To bound the probability of the event $E_{3i}$, we use a standard argument \cite{GP80}. More specifically, conditioned on $E^c_{1,i}$ and  $E^c_{2,i}$, the complement events of $E_{1i}$ and $E_{2i}$, respectively, we have that the state $\dv s[i]$ is jointly typical with $(\dv x_{2,i}(w_c,t_{i-1}), \dv v_i(w_c,t_{i-1},t_i))$. Then, for $\dv u_i(w_c,t_{i-1},t_i,w_1,j_{i})$ generated independently of $\dv s[i]$  given $\dv x_{2,i}(w_c,t_{i-1})$ and $\dv v_i(w_c,t_{i-1},t_i)$, with i.i.d. components drawn according to $P_{U|V,X_2}$, the probability that $\dv u_i(w_c,t_{i-1},t_i,w_1,j_{i})$ is jointly typical with $\dv s[i]$ given $\dv x_{2,i}(w_c,t_{i-1})$ and $\dv v_i(w_c,t_{i-1},t_i)$ is greater than $(1-\epsilon)2^{-n(I(U;S|V,X_2)+\epsilon)}$ for sufficiently large $n$.  There is a total of $J$ such $\dv u_i$'s in each bin. Conditioned on $E^c_{1i}$ and $E^c_{2i}$, the probability of the event $E_{3i}$, the probability that there is no such $\dv u_i$, is therefore bounded as
\begin{equation}
\text{Pr}(E_{3i}|E^c_{1i},E^c_{2i}) \leq [1-(1-\epsilon)2^{-n(I(U;S|V,X_2)+\epsilon)}]^{J}.
\label{BoundingProbabilityErrorEventE2}
\end{equation}
Taking the logarithm on both sides of \eqref{BoundingProbabilityErrorEventE2} and substituting $J$, we obtain that $\ln(\text{Pr}(E_{3i}|E^c_{1i},E^c_{2i})) \leq -(1-\epsilon)2^{n(\delta-1)\epsilon}$. Thus, $\text{Pr}(E_{3i}|E^c_{1i},E^c_{2i}) \rightarrow 0 \quad \text{as}\quad  n \rightarrow \infty$ and, so, by the union bound, $\text{Pr}(E_3|E^c_1,E^c_2) \rightarrow 0 \quad \text{as}\quad  n \rightarrow \infty$.

%----------------------
 \item  For the decoding of the common message $w_c$ at the receiver, let $E_4=\cup_{i=1}^{B}E_{4i}$ where $E_{4i}$ is the event that $\Big(\dv x_{2,i}(w_c,t_{i-1})$, $\dv u_i(w_c,t_{i-1},t_i,w_1,j^{\star}_{i})$, $\dv v_i(w_c,t_{i-1},t_i)$, $\dv y[i]\Big)$ is not jointly typical, i.e.,
\begin{align}
   E_4 &= \bigcup_{i=1}^{B} \Big\{\Big(\dv x_{2,i}(w_c,t_{i-1}), \dv u_i(w_c,t_{i-1},t_i,w_1,j^{\star}_{i}), \dv v_i(w_c,t_{i-1},t_i), \dv y[i]\Big) \notin \mc T_{\epsilon}^n(P_{X_2,U,V,Y})\Big\}.
\end{align}

Conditioned on $E^c_{1i}$, $E^c_{2i}$ and $E^c_{3i}$, the vectors $\dv s[i]$, $\dv x_{2,i}(w_c,t_{i-1})$, $\dv v_i(w_c,t_{i-1},t_i)$ and $\dv u_i(w_c,t_{i-1},t_i,w_1,j^{\star}_{i})$ are jointly typical and with $\dv x_1[i]$. Then, conditioned on $E^c_{1i}$, $E^c_{2i}$ and $E^c_{3i}$, the vectors $\dv s[i]$, $\dv x_{2,i}(w_c,t_{i-1})$, $\dv v_i(w_c,t_{i-1},t_i)$, $\dv u_i(w_c,t_{i-1},t_i,w_1,j^{\star}_{i})$ and $\dv y[i]$ are jointly typical by the Markov lemma \cite[p. 436]{CT91}, i.e., $\text{Pr}(E_{4i}|E^c_{1i},E^c_{2i},E^c_{3i}) \rightarrow 0 \quad \text{as}\quad  n \rightarrow \infty$. Thus, by the union bound over the $B$ blocks, $\text{Pr}(E_4|E^c_1,E^c_2,E^c_3) \rightarrow 0 \quad \text{as}\quad  n \rightarrow \infty$.

%----------------------
\item  For the decoding of the common message $w_c$ at the receiver, let $E_{5}$ be the event that $\dv x_{2,i}(w'_c,t_{i-1})$, $\dv u_i(w'_c,t_{i-1},t_i,w_1,j_{i})$, $\dv v_i(w'_c,t_{i-1},t_i)$ and $\dv y[i]$ are jointly typical for all $i=1,\hdots,B$ and some $w'_c \in [1:M_c]$, $w_1 \in [1:M_1]$, $t^B=(t_1,\hdots,t_B) \in [1:\hat{M}]^{B}$ and $j^B=(j_1,\hdots,j_B) \in [1:J]^B$ such that $w'_c \neq w_c$, i.e.,
\begin{align}
E_5 = \bigg\{ &\: \exists \: w'_c \in [1:M_c],\: w_1 \in [1:M_1],\: t^B=(t_1,\hdots,t_B) \in [1:\hat{M}]^{B},\: j^B \in [1:J]^B\:\text{s.t.:} \: w'_c \neq w_c,\nonumber\\
& \bigcap_{i=1}^{B} \Big\{\Big(\dv x_{2,i}(w'_c,t_{i-1}), \dv u_i(w'_c,t_{i-1},t_i,w_1,j_{i}), \dv v_i(w'_c,t_{i-1},t_i), \dv y[i]\Big) \in \mc T_{\epsilon}^n(P_{X_2,U,V,Y})\Big\}\bigg\}.
\end{align}

To bound the probability of the event $E_5$, define the following event for given $w'_c \in [1:M_c]$, $w_1 \in [1:M_1]$, $(t_{i-1},t_i) \in [1:\hat{M}]^2$ and $j_i \in [1:J]$ such that $w'_c \neq w_c$,
\begin{align*}
E_{5i}(w'_c,t_{i-1},t_i,w_1,j_i) = \Big\{ &\:\Big(\dv x_{2,i}(w'_c,t_{i-1}), \dv u_i(w'_c,t_{i-1},t_i,w_1,j_i), \dv v_i(w'_c,t_{i-1},t_i), \dv y[i]\Big) \in \mc T_{\epsilon}^n(P_{X_2,U,V,Y})\Big\}.
\end{align*}

Note that for $w'_c \neq w_c$ the vectors $\dv x_{2,i}(w'_c,t_{i-1})$,  $\dv v_i(w'_c,t_{i-1},t_i)$ and $\dv u_i(w'_c,t_{i-1},t_i,w_1,j_i)$ are generated independently of $\dv y[i]$. Hence, by the joint typicality lemma \cite[Lecture Note 2]{GK10}, we get
\begin{align}
\text{Pr}\Big(E_{5i}(w'_c,t_{i-1},t_i,w_1,j_i)|E_1^c,E_2^c,E_3^c,E_4^c\Big) & \leq 2^{-n[I(U,V,X_2;Y)-\epsilon]}.
\label{BoundingProbabilityIntermediaryEventE5_Step1}
\end{align}

Then, conditioned on the events $E_{1}^c$, $E_{2}^c$, $E_{3}^c$ and $E_{4}^c$, the probability of the event $E_5$ can be bounded as
\begin{align}
\text{Pr}(E_5|E_1^c,E_2^c,E_3^c,E_4^c) &= \text{Pr}\Big( \bigcup_{w'_c \neq w_c} \bigcup_{w_1 \: \in \: [1:M_1]} \bigcup_{t^B \: \in \: [1:\hat{M}]^{B}} \bigcup_{j^B \: \in \: [1:J]^{B}} \bigcap_{i=1}^{B} E_{5i}(w'_c,t_{i-1},t_i,w_1,j_i)|E_1^c,E_2^c,E_3^c,E^c_4\Big)\nonumber\\
&\stackrel{(a)}{\leq} \sum_{w'_c \neq w_c} \sum_{w_1 \in [1:M_1]} \sum_{t^B \: \in \: [1:\hat{M}]^{B}} \sum_{j^B \: \in \: [1:J]^{B}}  \text{Pr}\Big(\bigcap_{i=1}^{B} E_{5i}(w'_c,t_{i-1},t_i,w_1,j_i)|E_1^c,E_2^c,E_3^c,E_4^c \Big)\nonumber\\
&\stackrel{(b)}{=} \sum_{w'_c \neq w_c} \sum_{w_1 \in [1:M_1]} \sum_{t^B \: \in \: [1:\hat{M}]^{B}}  \sum_{j^B \: \in \: [1:J]^{B}} \prod_{i=1}^{B} \text{Pr}\Big(E_{5i}(w'_c,t_{i-1},t_i,w_1,j_i)|E_1^c,E_2^c,E_3^c,E^c_4\Big)\nonumber\\
&\leq \sum_{w'_c \neq w_c} \sum_{w_1 \in [1:M_1]} \sum_{t^B \: \in \: [1:\hat{M}]^{B}} \sum_{j^B \: \in \: [1:J]^{B}}  \prod_{i=2}^{B} \text{Pr}\Big(E_{5i}(w'_c,t_{i-1},t_i,w_1,j_i)|E_1^c,E_2^c,E_3^c,E_4^c\Big)\nonumber\\
&\stackrel{(c)}{\leq} \sum_{w'_c \neq w_c} \sum_{w_1 \in [1:M_1]} \sum_{t^B \: \in \: [1:\hat{M}]^{B}}  \sum_{j^B \: \in \: [1:J]^{B}} \prod_{i=2}^{B} 2^{-n\big[I(U,V,X_2;Y)-\epsilon\big]}\nonumber\\
&= \sum_{w'_c \neq w_c} \sum_{w_1 \in [1:M_1]} \sum_{t_B \: \in \: [1:\hat{M}]} \sum_{j_B \: \in \: [1:J]} 2^{n(B-1)\big[\hat{R}+\hat{\eta}\epsilon\big]}2^{-n(B-1)\big[I(U,V,X_2;Y)-I(U;S|V,X_2)-(\delta+1)\epsilon\big]}\nonumber\\
&\leq M_cM_1\hat{M}J2^{-n(B-1)\big[I(U,V,X_2;Y)-I(U;S|V,X_2)-\hat{R}-(\hat{\eta}+\delta+1)\epsilon\big]}\nonumber\\
&= 2^{-nB\big[\frac{B-1}{B}\big(I(U,V,X_2;Y)-I(U;S|V,X_2)-\hat{R}\big)-(R_c+R_1)-\frac{\hat{R}}{B}-\frac{I(U;S|V,X_2)}{B}+\big(\eta_c+\eta_1-\hat{\eta}-\delta-\frac{B-1}{B}\big)\epsilon\big]}
\label{BoundingProbabilityIntermediaryEventE5_Step2}
\end{align}
where: $(a)$ follows by the union bound; $(b)$ follows since the codebook is generated independently for each block $i \in [1:B]$ and the channel is memoryless; and $(c)$ follows by \eqref{BoundingProbabilityIntermediaryEventE5_Step1}.

The right hand side (RHS) of \eqref{BoundingProbabilityIntermediaryEventE5_Step2} tends to zero as $n \rightarrow \infty$ if
\begin{align}
  R_c+R_1 \leq \frac{B-1}{B}\big(I(U,V,X_2;Y)-I(U;S|V,X_2)-\hat{R}\big)-\frac{\hat{R}}{B}-\frac{I(U;S|V,X_2)}{B}.
  \label{BoundingProbabilityIntermediaryEventE5_Step3}
\end{align}

Finally, using \eqref{Condition__on__CoveringRate} to eliminate $\hat{R}$ from \eqref{BoundingProbabilityIntermediaryEventE5_Step3} and taking $B \rightarrow \infty$, we get $\text{Pr}(E_5|E^c_1,E^c_2,E^c_3,E^c_4) \rightarrow 0$ as long as
\begin{align}
R_c+R_1 &\leq I(U,V,X_2;Y)-I(U,V;S|X_2)\nonumber\\
        &= I(U,V,X_2;Y)-I(U,V,X_2;S),
\label{Condition__on__SumRate}
\end{align}
where the last equality follows since $X_2$ and $S$ are independent.

%----------------------
\item For the decoding of the individual message $w_1$ at the receiver, let $E_6=\cup_{i=1}^{B} E_{6i}$ where $E_{6i}$ is the event that $\dv x_{2,i}(w_c,t_{i-1})$, $\dv v_i(w_c,t_{i-1},t_i)$, $\dv u_i(w_c,t_{i-1},t_i,w_1,j^{\star}_{i})$ and $\dv y[i]$ are not jointly typical, i.e.,
\begin{align}
E_{6i} &= \Big\{\Big(\dv x_{2,i}(w_c,t_{i-1}), \dv v_i(w_c,t_{i-1},t_i), \dv u_i(w_c,t_{i-1},t_i,w_1,j^{\star}_{i}), \dv y[i]\Big) \notin \mc T_{\epsilon}^n(P_{X_2,V,U,Y})\Big\}.
\end{align}

From our analysis of the probability of the error event $E_4$, it is easy to see that, conditioned on $E^c_1$, $E^c_2$ and $E^c_3$, the event $E_{6i}$ has exponentially small probability. Thus, by the union bound over the $B$ blocks, $\text{Pr}(E_6|E^c_1,E^c_2,E^c_3) \longrightarrow 0$ as $n \longrightarrow \infty$, where $E_6=\cup_{i=1}^{B} E_{6i}$.

%----------------------
\item  For the decoding of the individual message $w_1$ at the receiver, let $E_7$ be the event that $\dv x_{2,i}(w_c,t_{i-1})$, $\dv u_i(w_c,t_{i-1},t_i,w'_1,j_{i})$, $\dv v_i(w_c,t_{i-1},t_i)$ and $\dv y[i]$ are jointly typical for all $i=1,\hdots,B$ and some $w'_1 \in [1:M_1]$, $t^B=(t_1,\hdots,t_B) \in [1:\hat{M}]^{B}$ and $j^B=(j_1,\hdots,j_B) \in [1:J]^B$ such that $w'_1 \neq w_1$, i.e.,
\begin{align}
E_7 = \bigg\{ &\: \exists \: w'_1 \in [1:M_1],\: t^B=(t_1,\hdots,t_B) \in [1:\hat{M}]^{B},\: j^B \in [1:J]^B\:\text{s.t.:} \: w'_1 \neq w_1,\nonumber\\
& \bigcap_{i=1}^{B} \Big\{\Big(\dv x_{2,i}(w_c,t_{i-1}), \dv u_i(w_c,t_{i-1},t_i,w'_1,j_{i}), \dv v_i(w_c,t_{i-1},t_i), \dv y[i]\Big) \in \mc T_{\epsilon}^n(P_{X_2,U,V,Y})\Big\}\bigg\}.
\end{align}

To bound the probability of the event $E_7$, define the following event for given $w'_1 \in [1:M_1]$, $(t_{i-1},t_i) \in [1:\hat{M}]^2$ and $j_i \in [1:J]$,
\begin{align*}
E_{7i}(w'_1,t_{i-1},t_i,j_i) = \Big\{ &\:\Big(\dv x_{2,i}(w_c,t_{i-1}), \dv u_i(w_c,t_{i-1},t_i,w'_1,j_i), \dv v_i(w_c,t_{i-1},t_i), \dv y[i]\Big) \in \mc T_{\epsilon}^n(P_{X_2,U,V,Y})\Big\}.
\end{align*}

Then, we have
\begin{align}
\text{Pr}(E_7|E_1^c,E_2^c,E_3^c,E_4^c,E_5^c,E_6^c) &= \text{Pr}\Big(\bigcup_{w'_1 \neq w_1} \bigcup_{t^B \: \in \: [1:\hat{M}]^{B}} \bigcup_{j^B \: \in \: [1:J]^{B}} \bigcap_{i=1}^{B} E_{7i}(w'_1,t_{i-1},t_i,j_i)|E_1^c,E_2^c,E_3^c,E^c_4,E_5^c,E^c_6\Big)\nonumber\\
&\stackrel{(d)}{\leq} \sum_{w'_1 \neq w_1} \sum_{t^B \: \in \: [1:\hat{M}]^{B}} \sum_{j^B \: \in \: [1:J]^{B}}  \text{Pr}\Big(\bigcap_{i=1}^{B} E_{7i}(w'_1,t_{i-1},t_i,j_i)|E_1^c,E_2^c,E_3^c,E_4^c,5^c,E^c_6 \Big)\nonumber\\
&\stackrel{(e)}{=} \sum_{w'_1 \neq w_1} \sum_{t^B \: \in \: [1:\hat{M}]^{B}}  \sum_{j^B \: \in \: [1:J]^{B}} \prod_{i=1}^{B} \text{Pr}\Big(E_{7i}(w'_1,t_{i-1},t_i,j_i)|E_1^c,E_2^c,E_3^c,E^c_4,5^c,E^c_6 \Big)\nonumber\\
&\leq \sum_{w'_1 \neq w_1} \sum_{t^B \: \in \: [1:\hat{M}]^{B}}  \sum_{j^B \: \in \: [1:J]^{B}} \prod_{i=2}^{B} \text{Pr}\Big(E_{7i}(w'_1,t_{i-1},t_i,j_i)|E_1^c,E_2^c,E_3^c,E^c_4,5^c,E^c_6 \Big)
\label{BoundingProbabilityIntermediaryEventE7_Step1}
\end{align}
where: $(d)$ follows by the union bound and $(e)$ follows since the codebook is generated independently for each block $i \in [1:B]$ and the channel is memoryless.

For $w'_1 \neq w_1$, the probability of the event $E_{7i}(w'_1,t_{i-1},t_i,j_i)$ conditioned on $E_1^c,E_2^c,E_3^c,E_4^c,E_5^c,E_6^c$ can be bounded as follows, depending on the values of $t_{i-1}$ and $t_i$:
\begin{itemize}
\item[i)] if $t_{i-1} \neq 1$ then $\Big(\dv u_i(w_c,t_{i-1},t_i,w'_1,j_i), \dv x_{2,i}(w_c,t_{i-1}), \dv v_i(w_c,t_{i-1},t_i)\Big)$ is generated independently of the output vector $\dv y[i]$ irrespective to the value of $t_i$, and so, by the joint typicality lemma \cite[Lecture Note 2]{GK10}
\begin{align}
\text{Pr}\Big(E_{7i}(w'_1,t_{i-1},t_i,j_i)|E_1^c,E_2^c,E_3^c,E_4^c,E_5^c,E_6^c\Big) & \leq 2^{-n[I(U,V,X_2;Y)-\epsilon]}.
\label{BoundingProbabilityIntermediaryEventE7_Step1}
\end{align}
\item[ii)] if $t_{i-1}=1$ and $t_i \neq 1$, then $\Big(\dv u_i(w_c,t_{i-1},t_i,w'_1,j_i), \dv v_i(w_c,t_{i-1},t_i)\Big)$ is generated independently of the output vector $\dv y[i]$  conditionally on $\dv x_{2,i}(w_c,t_{i-1})$; and, hence
\begin{align}
\text{Pr}\Big(E_{7i}(w'_1,t_{i-1},t_i,j_i)|E_1^c,E_2^c,E_3^c,E_4^c,E_5^c,E_6^c\Big) & \leq 2^{-n[I(U,V;Y|X_2)-\epsilon]}.
\label{BoundingProbabilityIntermediaryEventE7_Step2}
\end{align}
\item[iii)] if $t_{i-1}=1$ and $t_i=1$, then $\dv u_i(w_c,t_{i-1},t_i,w'_1,j_i)$ is generated independently of the output vector $\dv y[i]$ conditionally on  $\dv x_{2,i}(w_c,t_{i-1})$ and $\dv v_i(w_c,t_{i-1},t_i)$; and, hence
\begin{align}
\text{Pr}\Big(E_{7i}(w'_1,t_{i-1},t_i,j_i)|E_1^c,E_2^c,E_3^c,E_4^c,E_5^c,E_6^c\Big) & \leq 2^{-n[I(U;Y|V,X_2)-\epsilon]}.
\label{BoundingProbabilityIntermediaryEventE7_Step3}
\end{align}
\end{itemize}

Now, note that since $I(U,V;Y|X_2) \geq I(U;Y|V,X_2)$, if $w'_1 \neq w_1$ and $t_{i-1}=1$ the following holds irrespective to the value of $t_i$,
\begin{align}
\text{Pr}\Big(E_{7i}(w'_1,t_{i-1},t_i,j_i)|E_1^c,E_2^c,E_3^c,E_4^c,E_5^c,E_6^c\Big) & \leq 2^{-n[I(U;Y|V,X_2)-\epsilon]}.
\end{align}

Let $I_1 := I(U;Y|V,X_2)$ and $I_2 := I(U,V,X_2;Y)$. If the sequence $(t_1,\hdots,t_{B-1})$ has $k$ ones, we have
\begin{align}
\prod_{i=2}^{B} \text{Pr}\Big(E_{7i}(w'_1,t_{i-1},t_i,j_i)E_1^c,E_2^c,E_3^c,E_4^c,E_5^c,E_6^c\Big) & \leq 2^{-n[kI_1+(B-1-k)I_2-(B-1)\epsilon]}.
\label{BoundingProbabilityIntermediaryEventE7_Step3}
\end{align}
Continuing from \eqref{BoundingProbabilityIntermediaryEventE7_Step1}, we then bound the probability of the event $E_7$ as
\begin{align}
&\text{Pr}(E_7|E_1^c,E_2^c,E_3^c,E_4^c,,E_5^c,E_6^c) \nonumber\\
&\leq \sum_{w'_1 \neq w_1} \sum_{t^B \: \in \: [1:\hat{M}]^{B}} \sum_{j^B \: \in \: [1:J]^{B}} \prod_{i=2}^{B} \text{Pr}\Big(E_{7i}(w'_1,t_{i-1},t_i,j_i)|E_1^c,E_2^c,E_3^c,E^c_4,5^c,E^c_6 \Big)\nonumber\\
&= \sum_{w'_1 \neq w_1} \sum_{t_B \: \in \: [1:\hat{M}]} \sum_{j^B \: \in \: [1:J]^{B}} \sum_{t^{B-1} \: \in \: [1:\hat{M}]^{B-1}} \prod_{i=2}^{B} \text{Pr}\Big(E_{7i}(w'_1,t_{i-1},t_i,j_i)|E_1^c,E_2^c,E_3^c,E^c_4,5^c,E^c_6 \Big)\nonumber\\
&\leq \sum_{w'_1 \neq w_1} \sum_{t_B \: \in \: [1:\hat{M}]}  \sum_{j^B \: \in \: [1:J]^{B}} \sum_{k=0}^{B-1}\binom{B-1}{k}\: 2^{n(B-1-k)\big[\hat{R}+\hat{\eta}\epsilon\big]} 2^{-n\big[kI_1+(B-1-k)I_2-(B-1)\epsilon\big]}\nonumber\\
&\leq \sum_{w'_1 \neq w_1} \sum_{t_B \: \in \: [1:\hat{M}]}  \sum_{j^B \: \in \: [1:J]^{B}} \sum_{k=0}^{B-1}\binom{B-1}{k}\: 2^{n(B-1-k)\big[\hat{R}+\hat{\eta}\epsilon\big]} 2^{-n\big[kI_1+(B-1-k)I_2-(B-1)\epsilon\big]}\nonumber\\
&= \sum_{w'_1 \neq w_1} \sum_{t_B \: \in \: [1:\hat{M}]}  \sum_{j_B \: \in \: [1:J]} \sum_{j^{B-1} \: \in \: [1:J]^{B-1}} \sum_{k=0}^{B-1}\binom{B-1}{k}\: 2^{-n\big[kI_1+(B-1-k)(I_2-\hat{R})-(B-1-k)\hat{\eta}\epsilon-(B-1)\epsilon\big]}\nonumber\\
&= \sum_{w'_1 \neq w_1} \sum_{t_B \: \in \: [1:\hat{M}]}  \sum_{j_B \: \in \: [1:J]} \sum_{k=0}^{B-1}\binom{B-1}{k}\: 2^{n(B-1)\big[I(U;S|V,X_2)+\delta\epsilon\big]} 2^{-n\big[kI_1+(B-1-k)(I_2-\hat{R})-(B-1)(\hat{\eta}+1)\epsilon\big]}\nonumber\\
&= \sum_{w'_1 \neq w_1} \sum_{t_B \: \in \: [1:\hat{M}]}  \sum_{j_B \: \in \: [1:J]} \sum_{k=0}^{B-1}\binom{B-1}{k}\: 2^{-n\big[k\big(I_1-I(U;S|V,X_2)\big)+(B-1-k)\big(I_2-\hat{R}-I(U;S|V,X_2)\big)-(B-1)(\hat{\eta}+\delta+1)\epsilon\big]}\nonumber\\
&\leq \sum_{w'_1 \neq w_1} \sum_{t_B \: \in \: [1:\hat{M}]}  \sum_{j_B \: \in \: [1:J]} \sum_{k=0}^{B-1}\binom{B-1}{k}\: 2^{-n\big[(B-1)\min\big(I_1-I(U;S|V,X_2),\: I_2-\hat{R}-I(U;S|V,X_2)\big)-(B-1)(\hat{\eta}+\delta+1)\epsilon\big]}\nonumber\\
&\leq M_1\hat{M}J 2^{B} 2^{-n\big[(B-1)\min\big(I_1-I(U;S|V,X_2),\: I_2-\hat{R}-I(U;S|V,X_2)\big)-(B-1)(\hat{\eta}+\delta+1)\epsilon\big]}\nonumber\\
&= 2^{-nB\big[\frac{B-1}{B}\min\big(I_1-I(U;S|V,X_2),\: I_2-\hat{R}-I(U;S|V,X_2)\big)-R_1-\frac{\hat{R}}{B}-\frac{I(U;S|V,X_2)}{B}-\frac{1}{n}+\big(\eta_1-\frac{\hat{\eta}}{B}-\frac{\delta}{B}-\frac{(B-1)(\hat{\eta}+\delta+1)}{B}\big)\epsilon\big]}\nonumber\\
&= 2^{-nB\big[\frac{B-1}{B}\min\big(I_1-I(U;S|V,X_2),\: I_2-\hat{R}-I(U;S|V,X_2)\big)-R_1-\frac{\hat{R}}{B}-\frac{I(U;S|V,X_2)}{B}-\frac{1}{n}+\big(\eta_1-\delta-\hat{\eta}-\frac{B-1}{B}\big)\epsilon\big]}.
\label{BoundingProbabilityIntermediaryEventE7_Step4}
\end{align}

The right hand side (RHS) of \eqref{BoundingProbabilityIntermediaryEventE5_Step2} tends to zero as $n \rightarrow \infty$ if
\begin{align}
R_1 \leq \frac{B-1}{B}\big(\min\big(I_1-I(U;S|V,X_2),\: I_2-\hat{R}-I(U;S|V,X_2)\big)-\frac{\hat{R}}{B}-\frac{I(U;S|V,X_2)}{B}.
\label{BoundingProbabilityIntermediaryEventE7_Step5}
\end{align}

Finally, using \eqref{Condition__on__CoveringRate} to eliminate $\hat{R}$ from \eqref{BoundingProbabilityIntermediaryEventE5_Step3} and taking $B \rightarrow \infty$, we get $\text{Pr}(E_7|E^c_1,E^c_2,E^c_3,E^c_4,E^c_5,E^c_6) \rightarrow 0$ as long as
\begin{align}
R_1 &\leq I_1-I(U;S|V,X_2)\nonumber\\
    &= I(U;Y|V,X_2)-I(U;S|V,X_2)
\label{Condition__on__IndividualRate}
\end{align}
and
\begin{align}
R_1 &\leq I_2-I(V;S|X_2)-I(U;S|V,X_2)\nonumber\\
    &= I(U,V,X_2;Y)-I(U,V,X_2;S).
\label{Condition__on__SumRate__Redundant}
\end{align}

\end{itemize}
%----------------------

Finally, noting that the condition \eqref{Condition__on__SumRate__Redundant} is redundant as $R_c \geq 0$ in \eqref{Condition__on__SumRate}, we obtain that the probability of error tends to zero as $n \rightarrow \infty$ and $B \rightarrow \infty$ if
\begin{subequations}
\begin{align}
R_1 &\leq I(U;Y|V,X_2)-I(U;S|V,X_2)\\
R_c+R_1 &\leq I(U,V,X_2;Y)-I(U,V,X_2;S).
\end{align}
\end{subequations}
This completes the proof of achievability.

%==============================================================================
\subsubsection{Converse Part of Theorem~\ref{Theorem__CapacityRegionDiscreteMemorylessChannel}}

We prove that for any $(M_c,M_1,n,\epsilon)$ code consisting of a mapping $\phi_1:\mc W_c{\times}\mc W_1{\times}\mc S^n \longrightarrow \mc X^n_1$ at Encoder 1, a sequence of mappings $\phi_{2,i}: \mc W_c{\times}\mc S^{i-1} \longrightarrow \mc X_2$, $i=1,\hdots,n$, at Encoder 2, and a mapping $\psi : \mc Y^n \longrightarrow \mc W_c{\times}\mc W_1$ at the decoder with average error probability $P_e^n \rightarrow 0$ as $n \rightarrow 0$ and rates $R_c=n^{-1}\log_2M_c$ and $R_1=n^{-1}\log_2M_1$, there exist random variables $(V,U,X_1,X_2) \in {\mc V}{\times}{\mc U}{\times}{\mc X_1}{\times}{\mc X_2}$ with $U$ and $V$ satisfying \eqref{BoundsOnCardinalityOfAuxiliaryRandonVariables__CapacityRegion__DiscreteMemorylessChannel} such that the joint distribution $P_{S,V,U,X_1,X_2}$ is of the form
\begin{align} 
P_{S,V,U,X_1,X_2}=Q_SP_{X_2}P_{V|S,X_2}P_{U,X_1|V,S,X_2},
\end{align}
the marginal distribution of $S$ is $Q_S(s)$, i.e.,
\begin{align}
         \sum_{v,u,x_1,x_2}P_{S,V,U,X_1,X_2}(s,v,u,x_1,x_2)=Q_S(s)
\end{align}
and the rate pair $(R_c,R_1)$ satisfies \eqref{CapacityRegionDiscreteMemorylessChannel}.

Define the random variables
\begin{align}
\bar{V}_i &= (W_c,S^{i-1},Y^n_{i+1})\nonumber\\
\bar{U}_i &= (W_1,\bar{V}_i).
\label{Definition__RabdomVariables__OuterBound}
\end{align}
Observe that the random variables so defined satisfy 
\begin{align}
(S_i,\bar{U}_i,\bar{V}_i, X_{1,i},X_{2,i},Y_i) \in \mc P, \quad \forall i \in\{1,\hdots,n\}.
\label{MeasureRandomVariables__OuterBound}
\end{align}

We first prove the following auxiliary result.

\begin{lemma}\label{LemmaProofOuterBound}
The following inequalities hold:
\begin{align}
I(W_1;Y^n|W_c) - I(W_1;S^n|W_c) &\leq \sum_{i=1}^{n} I(\bar{U}_i;Y_i|\bar{V}_i,X_{2,i})-I(\bar{U}_i;S_i|\bar{V}_i,X_{2,i})\\
I(W_c,W_1;Y^n)-I(W_c,W_1;S^n) &\leq \sum_{i=1}^{n} I(\bar{U}_i,\bar{V}_i,X_{2i};Y_i)-I(\bar{U}_i,\bar{V}_i;S_i|X_{2i}) 
\end{align}
\end{lemma}

\begin{proof} % PROOF OF LEMMA
i) We show the first inequality in the lemma as follows.
{\allowdisplaybreaks
\begin{align}
I(W_1;&Y^n|W_c)-I(W_1;S^n|W_c) \\
       &= \sum_{i=1}^{n} I(W_1;Y_i|W_c,Y^n_{i+1})-I(W_1;S_i|W_c,S^{i-1})\\
       &= \sum_{i=1}^{n} I(W_1,S^{i-1};Y_i|W_c,Y^n_{i+1})-I(S^{i-1};Y_i|W_c,W_1,Y^n_{i+1})-I(W_1;S_i|W_c,S^{i-1})\\
       &= \sum_{i=1}^{n} I(W_1,S^{i-1};Y_i|W_c,Y^n_{i+1})-I(W_1;S_i|W_c,S^{i-1}) - \sum_{i=1}^{n} I(S^{i-1};Y_i|W_c,W_1,Y^n_{i+1})\\
       &\stackrel{(a)}{=} \sum_{i=1}^{n} I(W_1,S^{i-1};Y_i|W_c,Y^n_{i+1})-I(W_1;S_i|W_c,S^{i-1}) - \sum_{i=1}^{n} I(S_i;Y^n_{i+1}|W_c,W_1,S^{i-1})\\
       &= \sum_{i=1}^{n} I(W_1,S^{i-1};Y_i|W_c,Y^n_{i+1})-I(S_i;W_1,Y^n_{i+1}|W_c,S^{i-1})\\
       &= \sum_{i=1}^{n} I(W_1;Y_i|W_c,S^{i-1},Y^n_{i+1})+I(S^{i-1};Y_i|W_c,Y^n_{i+1})-I(S_i;Y^n_{i+1}|W_c,S^{i-1})-I(S_i;W_1|W_c,S^{i-1},Y^n_{i+1})\\
       &= \sum_{i=1}^{n} I(W_1;Y_i|W_c,S^{i-1},Y^n_{i+1})-I(S_i;W_1|W_c,S^{i-1},Y^n_{i+1})+\sum_{i=1}^{n} I(S^{i-1};Y_i|W_c,Y^n_{i+1})-\sum_{i=1}^{n} I(S_i;Y^n_{i+1}|W_c,S^{i-1})\\
       &\stackrel{(b)}{=} \sum_{i=1}^{n} I(W_1;Y_i|W_c,S^{i-1},Y^n_{i+1})-I(S_i;W_1|W_c,S^{i-1},Y^n_{i+1})\\
       &\stackrel{(c)}{=} \sum_{i=1}^{n} I(W_1;Y_i|W_c,S^{i-1},Y^n_{i+1},X_{2,i})-I(S_i;W_1|W_c,S^{i-1},Y^n_{i+1},X_{2,i})\\
       &\stackrel{(d)}{=} \sum_{i=1}^{n} I(\bar{U}_i;Y_i|\bar{V}_i,X_{2,i})-I(\bar{U}_i;S_i|\bar{V}_i,X_{2,i})
\label{MultiLetter__UpperBound__IndividualRate}
\end{align}}
where $(a)$ and $(b)$ follow from  Csisz\'ar and K\"orner's Sum Identities\cite{CK78}
\begin{align}
\label{CsiszarKorner__FirstTerm}
\sum_{i=1}^{n} I(S^{i-1};Y_i|W_c,W_1,Y^n_{i+1}) &= \sum_{i=1}^{n} I(S_i;Y^n_{i+1}|W_c,W_1,S^{i-1})\\
\sum_{i=1}^{n} I(S^{i-1};Y_i|W_c,Y^n_{i+1}) &= \sum_{i=1}^{n} I(S_i;Y^n_{i+1}|W_c,S^{i-1})
\label{CsiszarKorner__SecondTerm}
\end{align}
$(c)$ follows from the fact that $X_{2i}$ is a deterministic function of $(W_c,S^{i-1})$, and $(d)$ follows by the definition of the random variables $\bar{U}_i$ and $\bar{V}_i$ in \eqref{Definition__RabdomVariables__OuterBound}.

ii) Similarly, we show the second inequality in the lemma as follows.
{\allowdisplaybreaks
\begin{align}
I(W_c,W_1;&Y^n)-I(W_c,W_1;S^n) \\
       &= \sum_{i=1}^{n} I(W_c,W_1;Y_i|Y^n_{i+1})-I(W_c,W_1;S_i|S^{i-1})\\
       &= \sum_{i=1}^{n} I(W_c,W_1,S^{i-1};Y_i|Y^n_{i+1})-I(S^{i-1};Y_i|W_c,W_1,Y^n_{i+1})-I(W_c,W_1;S_i|S^{i-1})\\
       &= \sum_{i=1}^{n} I(W_c,W_1,S^{i-1};Y_i|Y^n_{i+1})-I(W_c,W_1;S_i|S^{i-1})-\sum_{i=1}^{n} I(S^{i-1};Y_i|W_c,W_1,Y^n_{i+1})\\
       &\stackrel{(e)}{=} \sum_{i=1}^{n} I(W_c,W_1,S^{i-1};Y_i|Y^n_{i+1})-I(W_c,W_1;S_i|S^{i-1})-\sum_{i=1}^{n} I(Y^n_{i+1};S_i|W_c,W_1,S^{i-1})\\
       &= \sum_{i=1}^{n} I(W_c,W_1,S^{i-1};Y_i|Y^n_{i+1})-I(W_c,W_1,Y^n_{i+1};S_i|S^{i-1})\\
       &=\sum_{i=1}^{n} I(W_c,W_1,S^{i-1};Y_i|Y^n_{i+1})-H(S_i|S^{i-1})+H(S_i|W_c,W_1,S^{i-1},Y^n_{i+1})\\
       &\stackrel{(f)}{=}\sum_{i=1}^{n} I(W_c,W_1,S^{i-1};Y_i|Y^n_{i+1})-H(S_i)+H(S_i|W_c,W_1,S^{i-1},Y^n_{i+1})\\
       &=\sum_{i=1}^{n} I(W_c,W_1,S^{i-1};Y_i|Y^n_{i+1})-I(W_c,W_1,S^{i-1},Y^n_{i+1};S_i)\\
       &\leq \sum_{i=1}^{n} I(W_c,W_1,S^{i-1},Y^n_{i+1};Y_i)-I(W_c,W_1,S^{i-1},Y^n_{i+1};S_i)\\
       &\stackrel{(g)}{=} \sum_{i=1}^{n} I(W_c,W_1,S^{i-1},Y^n_{i+1},X_{2i};Y_i)-I(W_c,W_1,S^{i-1},Y^n_{i+1},X_{2i};S_i)\\
       &\stackrel{(h)}{=} \sum_{i=1}^{n} I(\bar{U}_i,\bar{V}_i,X_{2i};Y_i)-I(\bar{U}_i,\bar{V}_i,X_{2i};S_i)
\label{MultiLetter__UpperBound__SumRate}
\end{align}}
where $(e)$ follows from Csisz\'ar and K\"orner's Sum Identity \eqref{CsiszarKorner__FirstTerm}; $(f)$ follows from the fact that the state $S^n$ is i.i.d.; $(g)$ follows from the fact that $X_{2i}$ is a deterministic function of $(W_c,S^{i-1})$, and $(h)$ follows by the definition of the random variables $\bar{U}_i$ and $\bar{V}_i$ in \eqref{Definition__RabdomVariables__OuterBound}.
\end{proof} % END PROOF OF LEMMA

We continue the proof of the converse. The decoder map $\psi$ recovers $(W_c,W_1)$ from $Y^n$ with vanishing average error probability $P^n_e$.  By Fano's inequality, we have
\begin{align}
        H(W_c,W_1|Y^n) \leq n\epsilon_n,
\label{FanoInequality}
\end{align}
where $\epsilon_n \rightarrow 0$ as $P_e^n \rightarrow 0$.

\noindent We can bound the individual rate as 
\begin{align}
 nR_1  &\leq H(W_1|W_c) \\
        &= I(W_1;Y^n|W_c)+H(W_1|Y^n,W_c)\\
	&\stackrel{(i)}{\leq} I(W_1;Y^n|W_c)+n\epsilon_n\\
	&\stackrel{(j)}{=} I(W_1;Y^n|W_c)-I(W_1;S^n|W_c)+n\epsilon_n\\
	&\stackrel{(k)}{=} \sum_{i=1}^{n} I(\bar{U}_i;Y_i|\bar{V}_i,X_{2,i})-I(\bar{U}_i;S_i|\bar{V}_i,X_{2,i})+n\epsilon_n
\end{align}
where $(i)$ follows by using \eqref{FanoInequality} and the fact that $H(W_1|W_c,Y^n) \leq H(W_c,W_1|Y^n)$; $(j)$ follows from the fact that the messages are independent of each other and of the state sequence; and $(k)$ follows by Lemma~\ref{LemmaProofOuterBound}.

\noindent Similarly, we can bound the sum rate as
\begin{align}
n(R_c+R_1)  &\leq H(W_c,W_1) \\
       &= I(W_c,W_1;Y^n)+H(W_c,W_1|Y^n)\\
       &\stackrel{(l)}{\leq} I(W_c,W_1;Y^n)+n\epsilon_n\\
       &\stackrel{(m)}{=} I(W_c,W_1;Y^n)-I(W_c,W_1;S^n)+n\epsilon_n\\
       &\stackrel{(n)}{=} \sum_{i=1}^{n} I(\bar{U}_i,\bar{V}_i,X_{2i};Y_i)-I(\bar{U}_i,\bar{V}_i,X_{2i};S_i),
\end{align}
where $(l)$ follows by \eqref{FanoInequality}; $(m)$ follows from the fact that the messages are independent of the state sequence; and $(n)$ follows by Lemma~\ref{LemmaProofOuterBound}.

From the above, we get that
\begin{align}
R_1 &\leq \frac{1}{n} \sum_{i=1}^{n} I(\bar{U}_i;Y_i|\bar{V}_i,X_{2,i})-I(\bar{U}_i;S_i|\bar{V}_i,X_{2,i})+\epsilon_n\nonumber\\
R_c+R_1 &\leq \frac{1}{n} \sum_{i=1}^{n} I(\bar{U}_i,\bar{V}_i,X_{2,i};Y_i)-I(\bar{U}_i,\bar{V}_i,X_{2,i};S_i)+\epsilon_n.
\label{MultiLetter__UpperBound__DiscreteChannel}
\end{align}
The statement of the converse follows now by applying to \eqref{MultiLetter__UpperBound__DiscreteChannel} the standard time-sharing argument and taking the limits of large $n$. This is shown briefly here. We introduce a random variable $T$ which is independent of $S$, and uniformly  distributed over $\{1,\cdots,n\}$. Set $S=S_T$, $\bar{U}=\bar{U}_T$, $\bar{V}=\bar{V}_T$, $X_1=X_{1,T}$, $X_2=X_{2,T}$, and $Y=Y_T$. Then, considering the first bound in \eqref{MultiLetter__UpperBound__DiscreteChannel}, we obtain
\begin{align}
\frac{1}{n} &\sum_{i=1}^{n} I(\bar{U}_i;Y_i|\bar{V}_i,X_{2,i})-I(\bar{U}_i;S_i|\bar{V}_i,X_{2,i})\nonumber\\
&= I(\bar{U};Y|\bar{V},X_2,T)-I(\bar{U};S|\bar{V},X_2,T)\nonumber\\
&= I(\bar{U},T;Y|\bar{V},X_2,T)-I(\bar{U},T;S|\bar{V},X_2,T).
\label{FirstTermUpperBound__WithTimeSharingVariable__DiscreteChannel}
\end{align}

\noindent Similarly, considering the second bound in \eqref{MultiLetter__UpperBound__DiscreteChannel}, we obtain
\begin{align}
\frac{1}{n} &\sum_{i=1}^{n} I(\bar{U}_i,\bar{V}_i,X_{2,i};Y_i)-I(\bar{U}_i,\bar{V}_i,X_{2,i};S_i)\nonumber\\
            &= I(\bar{U},\bar{V},X_2;Y|T)-I(\bar{U},\bar{V},X_2;S|T)\nonumber\\
	    &= I(T,\bar{U},\bar{V},X_2;Y)-I(T;Y)-I(T,\bar{U},\bar{V},X_2;S)+I(T;S)\nonumber\\
	    &\leq I(T,\bar{U},\bar{V},X_2;Y)-I(T,\bar{U},\bar{V},X_2;S).
\label{SecondTermUpperBound__WithTimeSharingVariable__DiscreteChannel}
\end{align}
The distribution on $(T,S,\bar{U},\bar{V},X_1,X_2,Y)$ from the given code is of the form
\begin{align}
P_{T,S,\bar{U},\bar{V},X_1,X_2,Y} &= Q_SP_TP_{X_2|T}P_{\bar{V}|X_2,S,T}P_{\bar{U},X_1|\bar{V},S,X_2,T}W_{Y|X_1,X_2,S}.
\label{MeasureOuterBound__with__TimeSharingRandomVariable}
\end{align}

\noindent Let us now define $U=(\bar{U},T)$ and $V=(\bar{V},T)$. Using \eqref{MultiLetter__UpperBound__DiscreteChannel}, \eqref{FirstTermUpperBound__WithTimeSharingVariable__DiscreteChannel} and \eqref{SecondTermUpperBound__WithTimeSharingVariable__DiscreteChannel}, we then get
\begin{align}
R_1 &\leq I(U;Y|V,X_2)-I(U;S|V,X_2)+\epsilon_n\nonumber\\
R_c+R_1 &\leq I(U,V,X_2;Y)-I(U,V,X_2;S)+\epsilon_n,
\end{align}
where the distribution on $(S,U,V,X_1,X_2,Y)$, obtained by marginalizing \eqref{MeasureOuterBound__with__TimeSharingRandomVariable} over the time sharing random variable $T$, satisfies $(S,U,V,X_1,X_2,Y) \in \mc P$.

So far we have shown that, for a given sequence of $(\epsilon_n,n,R_c,R_1)-$codes with $\epsilon_n$ going to zero as $n$ goes to infinity, there exist random variables $(S,U,V,X_1,X_2,Y) \in \mc P$ such that the rate pair $(R_c,R_1)$ essentially satisfies the inequalities in \eqref{CapacityRegionDiscreteMemorylessChannel}, i.e., $(R_c,R_1) \in \mc C$.

\noindent This completes the proof of the converse part and of Theorem~\ref{Theorem__CapacityRegionDiscreteMemorylessChannel}.

%%%%%%%%%%%%%%%%%%%%%%%%%%%%%%%%%%%%%%%%%%%%%%%%%%%%%%%%%%%%%%%%%%%%%%%%%%%%%%%%%%%%%%%%
\renewcommand{\theequation}{C-\arabic{equation}}
\setcounter{equation}{0}  % reset counter
\subsection{Proof of Theorem~\ref{Theorem__WynerZivBinningOptimality}}\label{appendixTheorem__WynerZivBinningOptimality}

The transmission takes place in $B$ blocks. The common message $W_c$ is divided into $B$ blocks $w_{c,1},\hdots,w_{c,B}$ of $nR_c$ bits each, and the individual messages $W_1$ is divided into $B$ blocks $w_{1,1},\hdots,w_{1,B}$ of $nR_1$ bits each. For convenience, we let $w_{c,B}=w_{1,B}=1$ (a default value). We thus have $B_{W_c}=n(B-1){R_c}$, $B_{W_1}=n(B-1){R_1}$, $N=nB$, $R_{W_c}=B_{W_c}/N=R_c{\cdot}(B-1)/B$ and $R_{W_1}=B_{W_1}/N=R_1{\cdot}(B-1)/B$, where $B_{W_c}$ is the number of common message bits, $B_{W_1}$ is the number of individual message bits, $N$ is the number of channel uses and $R_{W_c}$ and $R_{W_1}$ are the overall rates of the common and individual messages, respectively. For fixed $n$, the average rate pair $(R_{W_c}, R_{W_1})$ over $B$ blocks can be made as close to $(R_c,R_1)$ as desired by making $B$ large.

\noindent \textbf{Codebook Generation:} Fix a measure $P_{S,U,V,X_1,X_2,Y} \in \mc P$. Fix $\epsilon > 0$ and denote $M_c = 2^{n[R_c-\eta_c\epsilon]}$, $M_1 = 2^{n[R_1-\eta_1\epsilon]}$, $M_0 = 2^{n[R_0+\eta_0\epsilon]}$, $\hat{M} = 2^{n[\hat{R}+\hat{\eta}\epsilon]}$, $J=2^{n[I(U;S|V,X_2)+\delta_U\epsilon]}$.

\begin{itemize}
\item[1)] We generate $M_cM_0$ independent and identically distributed (i.i.d.) codewords $\dv x_2(w_c,s)$ indexed by $w_c=1,\hdots,R_c$, $s=1,\hdots,M_0$, each with i.i.d. components drawn according to $P_{X_2}$.
\item[2)] For each codeword $\dv x_2(w_c,s)$,  we generate $\hat{M}$ independent and identically distributed (i.i.d.) codewords $\dv v(w_c,s,z)$ indexed by $z=1,\hdots,\hat{M}$, each with i.i.d. components drawn according to $P_{V|X_2}$.
\item[3)] For each codeword $\dv x_2(w_c,s)$, for each codeword $\dv v(w_c,s,z)$, we generate a collection of $JM_1$ i.i.d. codewords $\{\dv u(w_c,s,z,w_1,j)\}$ indexed by $w_1=1,\hdots,M_1$, $j=1,\hdots,J$, each with i.i.d. components draw according to $P_{U|V,X_2}$.
\item[4)] Randomly partition the set $\{1,\hdots,\hat{M}\}$ into $M_0$ cells $\mc C_s$, $s \in [1,M_0]$.
\end{itemize}

\textbf{Encoding:} Suppose that a common message $W_c=w_c$ and an individual message $W_1=w_1$ are to be transmitted. As we mentioned previously, message $w_c$ is divided into $B$ blocks $w_{c,1},\hdots,w_{c,B}$ and message $w_1$ is divided into $B$ blocks $w_{1,1},\hdots,w_{1,B}$, with $(w_{c,i},w_{1,i})$ the pair messages sent in block $i$. We denote by $\dv s[i]$ the channel state in block $i$, $i=1,\hdots,B$. For convenience, we let $\dv s[0]=\phi$ and $z_0=1$ (a default value), and $s_0$ the index of the cell containing $z_0$, i.e., $z_0 \in \mc C_{s_0}$ . The encoding at the beginning of the block $i$, $i=1,\hdots,B$, is as follows.

\noindent Encoder $2$, which has learned the state sequence $\dv s[i-1]$, knows $s_{i-2}$ and looks for a compression index $z_{i-1} \in [1,\hat{M}]$ such that $\dv v(w_{c,i-1},s_{i-2},z_{i-1})$ is strongly jointly typical with $\dv s[i-1]$ and $\dv x_2(w_{c,i-1},s_{i-2})$. If there is no such index or the observed state $\dv s[i-1]$ is not typical, $z_{i-1}$ is set to $1$ and an error is declared. If there is more than one such index $z_{i-1}$, choose the smallest. One can show that the probability of error of this event is arbitrarily small provided that $n$ is large and
\begin{align}
\hat{R} &> I(V;S|X_2).
\end{align}

Encoder 2 then transmits the vector $\dv x_2(w_{c,i},s_{i-1})$, where $s_{i-1}$ is such that $z_{i-1} \in \mc C_{s_{i-1}}$.

\noindent Encoder 1 obtains $\dv x_2(w_{c,i},s_{i-1})$ similarly. It then finds the smallest compression index $z_i \in [1,\hat{M}]$ such that $\dv v(w_{c,i},s_{i-1},z_i)$ is strongly jointly typical with $\dv s[i]$ and $\dv x_2(w_{c,i},s_{i-1})$. Again, if there is no such index or the observed state $\dv s[i]$ is not typical, $z_i$ is set to $1$ and an error is declared. Let $s_i \in [1,M_0]$ such that $z_i \in \mc C_{s_i}$. Next, Encoder 1 looks for the smallest $j_{i}$ such that $\dv u(w_{c,i},s_{i-1},z_i,w_{1,i},j_{i})$ is jointly typical with $\dv s[i]$, $\dv x_2(w_{c,i},s_{i-1})$ and $\dv v(w_{c,i},s_{i-1},z_i)$. Denote this $j_{i}$ by $j^{\star}_{i}=j(\dv s[i],w_{c,i},s_{i-1},z_i,w_{1,i})$. If such $j^{\star}_{i}$ is not found, an error is declared and $j(\dv s[i],w_{c,i},s_{i-1},z_i,w_{1,i})$ is set to $j_{i}=J$. Encoder 1 then transmits a vector $\dv x_1[i]$ which is drawn i.i.d. conditionally given $\dv s[i]$,  $\dv u(w_{c,i},s_{i-1},z_i,w_{1,i},j^{\star}_{i})$, $\dv v(w_{c,i},s_{i-1},z_i)$ and $\dv x_2(w_{c,i},s_{i-1})$ (using the conditional measure $P_{X_1|S,U,V,X_2}$ induced by  $P_{S,U,V,X_1,X_2,Y} \in \mc P$).

\textbf{Decoding:} Let $\dv y[i]$ denote the information received at the receiver at block $i$, $i=1,\hdots,B$. The receiver collects these information until the last block of transmission is completed. The decoder then performs Willem's backward decoding \cite{W82}, by first decoding the pair $(w_{c,B-1},w_{1,B-1})$ from $\dv y[B-1]$.

\textit{1)Decoding in Block $B-1$:}

The decoding of the pair $(w_{c,B-1},w_{1,B-1})$ is performed in four steps, as follows.

\noindent \underline{\textit{Step (a):}} The decoder knows $w_{c,B}=1$ and looks for the unique cell index $\hat{s}_{B-1}$ such that the vector $\dv x_2(w_{c,B},\hat{s}_{B-1})$ is jointly typical with $\dv y[B]$.  The decoding operation in this step incurs small probability of error as long as $n$ is sufficiently large and
\begin{align}
R_0 &< I(X_2;Y).
\label{DecodingOfCellIndex}
\end{align}

\noindent \underline{\textit{Step (b):}} The decoder now knows $\hat{s}_{B-1}$ (i.e., the index of the cell in which the compression index $z_{B-1}$ lies). It then decodes message $w_{c,B-1}$ by looking for the unique $\hat{w}_{c,B-1}$ such that $\dv x_2(\hat{w}_{c,B-1},s_{B-2})$,  $\dv v(\hat{w}_{c,B-1},s_{B-2},z_{B-1})$, $\dv u(\hat{w}_{c,B-1},s_{B-2},z_{B-1},w_{1,B-1},j_{B-1})$ and $\dv y[B-1]$ are jointly typical for some $s_{B-2} \in [1,M_0]$, $w_{1,B-1} \in [1,M_1]$, $j_{B-1} \in [1,J]$ and $z_{B-1} \in \mc C_{\hat{s}_{B-1}}$. One can show that the decoder obtains the correct $w_{c,B-1}$ as long as $n$ and $B$ are large and
\begin{align}
R_0+(\hat{R}-R_0)+ R_c + R_1 &\leq I(U,V,X_2;Y)-I(U;S|V,X_2).
\label{Constraint__On__SumRate}
\end{align}

\noindent \underline{\textit{Step (c):}} The decoder knows $\hat{w}_{c,B-1}$ and can again obtain the correct $s_{B-2}$ if $n$ is large and \eqref{DecodingOfCellIndex} is true. This is accomplished by looking for the unique $\hat{s}_{B-2}$ such that the vector $\dv x_2(\hat{w}_{c,B-1},\hat{s}_{B-2})$ is jointly typical with $\dv y[B-1]$.

\noindent \underline{\textit{Step (d):}} Finally, the decoder, which now knows message  $\hat{w}_{c,B-1}$ and the cell index $\hat{s}_{B-2}$ (but not the exact compression index $z_{B-1}$), estimates $w_{1,B-1}$ using $\dv y[B-1]$. It declares that $\hat{w}_{1,B-1}$ was sent if there exists a unique $\hat{w}_{1,B-1}$ such that $\dv x_2(\hat{w}_{c,B-1},\hat{s}_{B-2})$, $\dv v(\hat{w}_{c,B-1},\hat{s}_{B-2},z'_{B-1})$, $\dv u(\hat{w}_{c,B-1},\hat{s}_{B-2},z'_{B-1},\hat{w}_{1,B-1},j_{B-1})$ and $\dv y[B-1]$ are jointly typical for some $z'_{B-1} \in \mc C_{\hat{s}_{B-1}}$ and $j_{B-1} \in [1,J]$.
\begin{itemize}
\item If $z'_{B-1} = z_{B-1}$, the decoder finds the correct $w_{1,b-1}$ for sufficiently large $n$ if
\begin{align}
R_1 &\leq I(U;Y|V,X_2)-I(U;S|V,X_2).
\label{Constraint1__On__IndividualRate}
\end{align}
\item If $z'_{B-1} \neq z_{B-1}$, the decoder finds the correct $w_{1,b-1}$ for sufficiently large $n$ if
\begin{align}
(\hat{R}-R_0)+R_1 &\leq I(U,V;Y|X_2)-I(U;S|V,X_2).
\label{Constraint2__On__IndividualRate}
\end{align}
\end{itemize}

\textit{2) Decoding in Block $b$, $b=B-1,B-2,\hdots,2$:}

Next, for $b$ ranging from $B-1$ to $2$, the decoding of the pair $(w_{c,b-1},w_{1,b-1})$ is performed similarly, in five steps, by using the information $\dv y[b]$ received in block $b$ and the information $\dv y[b-1]$ received in block $b-1$. More specifically, this is done as follows.

\noindent \underline{\textit{Step (a):}} The decoder knows $w_{c,b}$ and looks for the unique cell index $\hat{s}_{b-1}$ such that the vector $\dv x_2(w_{c,b},\hat{s}_{b-1})$ is jointly typical with $\dv y[b]$. The decoding error in this step is small for sufficiently large $n$ if \eqref{DecodingOfCellIndex} is true.

\noindent \underline{\textit{Step (b):}} The decoder knows $\hat{s}_{b-1}$ and decodes message $w_{c,b-1}$ from $\dv y[b]$. It looks for the unique $\hat{w}_{c,b-1}$ such that $\dv x_2(\hat{w}_{c,b-1},s_{b-2})$,  $\dv v(\hat{w}_{c,b-1},s_{b-2},z_{b-1})$, $\dv u(\hat{w}_{c,b-1},s_{b-2},z_{b-1},w_{1,b-1},j_{b-1})$ and $\dv y[b-1]$ are jointly typical for some $s_{b-2} \in [1,M_0]$, $w_{1,b-1} \in [1,M_1]$, $j_{b-1} \in [1,J]$ and $z_{b-1} \in \mc C_{\hat{s}_{b-1}}$. One can show that the decoding error in this step is small for sufficiently large $n$ if \eqref{Constraint__On__SumRate} is true.

\noindent \underline{\textit{Step (c):}} The decoder knows $\hat{w}_{c,b-1}$ and obtains $\hat{s}_{b-2}$ by looking for the unique $\hat{s}_{b-2}$ such that the vector $\dv x_2(\hat{w}_{c,b-1},\hat{s}_{b-2})$ is jointly typical with $\dv y[b-1]$. For sufficiently large $n$, the decoder obtains the correct $s_{b-2}$ with high probability if \eqref{DecodingOfCellIndex} is true.

\noindent \underline{\textit{Step (d):}} Finally, the decoder, which now knows message  $\hat{w}_{c,b-1}$ and the cell index $\hat{s}_{b-2}$ (but not the exact compression index $z_{b-1}$), estimates message $w_{1,b-1}$ using $\dv y[b-1]$. It declares that $\hat{w}_{1,b-1}$ was sent if there exists a unique $\hat{w}_{1,b-1}$ such that $\dv x_2(\hat{w}_{c,b-1},\hat{s}_{b-2})$, $\dv v(\hat{w}_{c,b-1},\hat{s}_{b-2},z'_{b-1})$, $\dv u(\hat{w}_{c,b-1},\hat{s}_{b-2},z'_{b-1},\hat{w}_{1,b-1},j_{b-1})$ and $\dv y[b-1]$ are jointly typical for some $z'_{b-1} \in \mc C_{\hat{s}_{b-1}}$ and $j_{b-1} \in [1,J]$.
\begin{itemize}
\item If $z'_{b-1} = z_{b-1}$, the decoder finds the correct $w_{1,b-1}$ for sufficiently large $n$ if \eqref{Constraint1__On__IndividualRate} is true.
\item If $z'_{b-1} \neq z_{b-1}$, the decoder finds the correct $w_{1,b-1}$ for sufficiently large $n$ if \eqref{Constraint2__On__IndividualRate} is true.
\end{itemize}

\textbf{Fourier-Motzkin Elimination:} From the above, we get that the error probability is small provided that $n$ is large and
\begin{subequations}
\begin{align}
R_0 &< I(X_2;Y)\\
\hat{R} &> I(V;S|X_2)\\
R_1 &\leq I(U;Y|V,X_2)-I(U;S|V,X_2)\\
(\hat{R}-R_0)+R_1 &\leq I(U,V;Y|X_2)-I(U;S|V,X_2)\\
R_c + R_1 + \hat{R} &\leq I(U,V,X_2;Y)-I(U;S|V,X_2).
\end{align}
\label{Theorem2RateRegion__Step1}
\end{subequations}
We now apply Fourier-Motzkin Elimination (FME) to project out  $R_0$ and $\hat{R}$ from \eqref{Theorem2RateRegion__Step1}. Projecting out $R_0$ from \eqref{Theorem2RateRegion__Step1}, we get
\begin{subequations}
\begin{align}
\hat{R} &> I(V;S|X_2)\\
R_1 &\leq I(U;Y|V,X_2)-I(U;S|V,X_2)\\
\label{Theorem2RateRegion__Step2__Ineq3}
\hat{R}+R_1 &\leq I(U,V,X_2;Y)-I(U;S|V,X_2)\\
R_c + R_1 + \hat{R} &\leq I(U,V,X_2;Y)-I(U;S|V,X_2).
\label{Theorem2RateRegion__Step2__Ineq4}
\end{align}
\label{Theorem2RateRegion__Step2}
\end{subequations}
Note that the inequality \eqref{Theorem2RateRegion__Step2__Ineq3} can be implied by \eqref{Theorem2RateRegion__Step2__Ineq4} since $R_c \geq 0$; and, so, is redundant in \eqref{Theorem2RateRegion__Step2}. Finally, projecting out $\hat{R}$ from the remaining system, we obtain
\begin{align}
R_1 &\leq I(U;Y|V,X_2)-I(U;S|V,X_2)\\
R_c + R_1 &\leq I(U,V,X_2;Y)-I(U,V,X_2;S).
\label{Theorem2RateRegion__Step4}
\end{align}

This completes the proof of Theorem~\ref{Theorem__WynerZivBinningOptimality}.

%%%%%%%%%%%%%%%%%%%%%%%%%%%%%%%%%%%%%%%%%%%%%%%%%%%%%%%%%%%%%%%%%%%%%%%%%%%%%%%%%%%%%%%%
\renewcommand{\theequation}{D-\arabic{equation}}
\setcounter{equation}{0}  % reset counter
\subsection{Proof of Corollary~\ref{Corollary__EquivalentCharacterizationCapacityRegionDiscreteMemorylessChannel}}\label{appendixCorollary__EquivalentCharacterizationCapacityRegionDiscreteMemorylessChannel}

\subsubsection{Converse Part}

Investigating the proof of Theorem~\ref{Theorem__CapacityRegionDiscreteMemorylessChannel} in Appendix~\ref{appendixTheorem__CapacityRegionDiscreteMemorylessChannel}, it can be seen that the auxiliary random variables $U$ and $V$ satisfy tacitly the condition
\begin{align}
I(V,X_2;Y)-I(V,X_2;S) &\geq 0.
\label{ConverseProof__ConstraintOuterBound__CompressionIndexDecoding}
\end{align}
This can be seen by noticing that (with the notation of Appendix~\ref{appendixTheorem__CapacityRegionDiscreteMemorylessChannel})
\begin{align}
I(W_1;Y^n|W_c) &= \sum_{i=1}^{n} I(\bar{U}_i;Y_i|\bar{V}_i,X_{2,i})-I(\bar{U}_i;S_i|\bar{V}_i,X_{2,i})\\
I(W_c,W_1;Y^n) &\leq \sum_{i=1}^{n} I(\bar{U}_i,\bar{V}_i,X_{2i};Y_i)-I(\bar{U}_i,\bar{V}_i,X_{2i};S_i).
\end{align}
and then observing that $I(W_1;Y^n|W_c) \leq I(W_c,W_1;Y^n)$, which together yield
\begin{align}
\sum_{i=1}^{n} I(\bar{V}_i,X_{2i};Y_i)-I(\bar{V}_i,X_{2i};S_i) &\geq 0;
\end{align}
and, so, after standard single-letterization, the condition \eqref{ConverseProof__ConstraintOuterBound__CompressionIndexDecoding}.

\subsubsection{Direct Part}

The codebook generation and the encoding process remain exactly as in the proof of Theorem~\ref{Theorem__WynerZivBinningOptimality} in Appendix~\ref{appendixTheorem__WynerZivBinningOptimality}. The decoding at the receiver is modified in a way to get the compression indices decoded uniquely, as follows (with the notation of Appendix~\ref{appendixTheorem__WynerZivBinningOptimality}). 

\textbf{Decoding:} Let $\dv y[i]$ denote the information received at the receiver at block $i$, $i=1,\hdots,B$. The receiver collects these information until the last block of transmission is completed. The decoder then performs Willem's backward decoding \cite{W82}, by first decoding the pair $(w_{c,B-1},w_{1,B-1})$ from $\dv y[B-1]$.

\textit{1) Decoding in Block $B-1$:}

The decoding of the pair $(w_{c,B-1},w_{1,B-1})$ is performed in five steps, as follows.

\noindent \underline{\textit{Step (a):}} The decoder knows $w_{c,B}=1$ and looks for the unique cell index $\hat{s}_{B-1}$ such that the vector $\dv x_2(w_{c,B},\hat{s}_{B-1})$ is jointly typical with $\dv y[B]$.  This decoding operation incurs small probability of error as long as $n$ is sufficiently large and
\begin{align}
R_0 &< I(X_2;Y).
\label{DecodingOfCellIndex__ProofTheorem2}
\end{align}

\noindent \underline{\textit{Step (b):}} The decoder now knows $\hat{s}_{B-1}$ (i.e., the index of the cell in which the compression index $z_{B-1}$ lies). It then decodes message $w_{c,B-1}$ by looking for the unique $\hat{w}_{c,B-1}$ such that $\dv x_2(\hat{w}_{c,B-1},s_{B-2})$,  $\dv v(\hat{w}_{c,B-1},s_{B-2},z_{B-1})$, $\dv u(\hat{w}_{c,B-1},s_{B-2},z_{B-1},w_{1,B-1},j_{B-1})$ and $\dv y[B-1]$ are jointly typical for some $s_{B-2} \in [1,M_0]$, $w_{1,B-1} \in [1,M_1]$, $j_{B-1} \in [1,J]$ and $z_{B-1} \in \mc C_{\hat{s}_{B-1}}$. One can show that the decoder obtains the correct $w_{c,B-1}$ as long as $n$ and $B$ are large and
\begin{align}
R_0+(\hat{R}-R_0)+ R_c + R_1 &\leq I(U,V,X_2;Y)-I(U;S|V,X_2).
\label{Constraint__On__SumRate}
\end{align}

\noindent \underline{\textit{Step (c):}} The decoder knows $\hat{w}_{c,B-1}$ and can again obtain the correct $s_{B-2}$ if $n$ is large and \eqref{DecodingOfCellIndex__ProofTheorem2} is true. This is accomplished by looking for the unique $\hat{s}_{B-2}$ such that the vector $\dv x_2(\hat{w}_{c,B-1},\hat{s}_{B-2})$ is jointly typical with $\dv y[B-1]$.

\noindent \underline{\textit{Step (d):}} The decoder calculates a set $\mc L(\dv y[B-1])$ of $z_{B-1}$ such that $z_{B-1} \in \mc L(\dv y[B-1])$ if $\dv v(\hat{w}_{c,B-1},\hat{s}_{B-2},z_{B-1})$, $\dv x_2(\hat{w}_{c,B-1},\hat{s}_{B-2})$, $\dv y[B-1]$ are jointly typical. It then declares that $z_{B-1}$ was sent in block $B-1$ if
\begin{align}
\hat{z}_{B-1} \in \mc C_{\hat{s}_{B-1}} \cap \mc L(\dv y[B-1]).
\end{align}
One can show that $\hat{z}_{B-1}=z_{B-1}$ with arbitrarily high probability provided that $n$ is sufficiently large and
\begin{align}
\hat{R} &< I(V;Y|X_2)+R_0.
\label{Constraint__Of__NonNegativity}
\end{align}
\noindent \underline{\textit{Step (e):}} Finally, the decoder, which now knows message  $\hat{w}_{c,B-1}$, the cell index $\hat{s}_{B-2}$ and the compression index $z_{B-1} \in \mc C_{\hat{s}_{B-1}}$, estimates $w_{1,B-1}$ using $\dv y[B-1]$. It declares that $\hat{w}_{1,B-1}$ was sent if there exists a unique $\hat{w}_{1,B-1}$ such that $\dv x_2(\hat{w}_{c,B-1},\hat{s}_{B-2})$, $\dv v(\hat{w}_{c,B-1},\hat{s}_{B-2},\hat{z}_{B-1})$,
$\dv u(\hat{w}_{c,B-1},\hat{s}_{B-2},\hat{z}_{B-1},\hat{w}_{1,B-1},j_{B-1})$ and $\dv y[B-1]$ are jointly typical for some $j_{B-1} \in [1,J]$. One can show that the decoder obtains the correct $w_{1,B-1}$ as long as $n$ is large and
\begin{align}
R_1 &\leq I(U;Y|V,X_2)-I(U;S|V,X_2).
\label{Constraint__On__IndividualRate}
\end{align}

\textit{2) Decoding in Block $b$, $b=B-1,B-2,\hdots,2$:}

Next, for $b$ ranging from $B-1$ to $2$, the decoding of the pair $(w_{c,b-1},w_{1,b-1})$ is performed similarly, in five steps, by using the information $\dv y[b]$ received in block $b$ and the information $\dv y[b-1]$ received in block $b-1$. More specifically, this is done as follows.

\noindent \underline{\textit{Step (a):}} The decoder knows $w_{c,b}$ and looks for the unique cell index $\hat{s}_{b-1}$ such that the vector $\dv x_2(w_{c,b},\hat{s}_{b-1})$ is jointly typical with $\dv y[b]$. The decoding error in this step is small for sufficiently large $n$ if \eqref{DecodingOfCellIndex__ProofTheorem2} is true.

\noindent \underline{\textit{Step (b):}} The decoder knows $\hat{s}_{b-1}$ and decodes message $w_{c,b-1}$ from $\dv y[b]$. It looks for the unique $\hat{w}_{c,b-1}$ such that $\dv x_2(\hat{w}_{c,b-1},s_{b-2})$,  $\dv v(\hat{w}_{c,b-1},s_{b-2},z_{b-1})$, $\dv u(\hat{w}_{c,b-1},s_{b-2},z_{b-1},w_{1,b-1},j_{b-1})$ and $\dv y[b-1]$ are jointly typical for some $s_{b-2} \in [1,M_0]$, $w_{1,b-1} \in [1,M_1]$, $j_{b-1} \in [1,J]$ and $z_{b-1} \in \mc C_{\hat{s}_{b-1}}$. One can show that the decoding error in this step is small for sufficiently large $n$ if \eqref{Constraint__On__SumRate} is true.

\noindent \underline{\textit{Step (c):}} The decoder knows $\hat{w}_{c,b-1}$ and obtains $\hat{s}_{b-2}$ by looking for the unique $\hat{s}_{b-2}$ such that the vector $\dv x_2(\hat{w}_{c,b-1},\hat{s}_{b-2})$ is jointly typical with $\dv y[b-1]$. For sufficiently large $n$, the decoder obtains the correct $s_{b-2}$ with high probability if \eqref{DecodingOfCellIndex__ProofTheorem2} is true.

\noindent \underline{\textit{Step (d):}} The decoder calculates a set $\mc L(\dv y[b-1])$ of $z_{b-1}$ such that $z_{b-1} \in \mc L(\dv y[b-1])$ if $\dv v(\hat{w}_{c,b-1},\hat{s}_{b-2},z_{b-1})$, $\dv x_2(\hat{w}_{c,b-1},\hat{s}_{b-2})$, $\dv y[b-1]$ are jointly typical. It then declares that $z_{b-1}$ was sent in block $b-1$ if
\begin{align}
\hat{z}_{b-1} \in \mc C_{\hat{s}_{b-1}} \cap \mc L(\dv y[b-1]).
\end{align}
One can show that, for large $n$, $\hat{z}_{b-1}=z_{b-1}$ with arbitrarily high probability provided that \eqref{Constraint__Of__NonNegativity} is true.

\noindent \underline{\textit{Step (e):}} Finally, the decoder knows message  $\hat{w}_{c,b-1}$, the cell index $\hat{s}_{b-2}$ and the compression index $z_{b-1} \in \mc C_{\hat{s}_{b-1}}$, and estimates $w_{1,b-1}$ using $\dv y[b-1]$. It declares that $\hat{w}_{1,b-1}$ was sent if there exists a unique $\hat{w}_{1,b-1}$ such that $\dv x_2(\hat{w}_{c,b-1},\hat{s}_{b-2})$, $\dv v(\hat{w}_{c,b-1},\hat{s}_{b-2},\hat{z}_{b-1})$, $\dv u(\hat{w}_{c,b-1},\hat{s}_{b-2},\hat{z}_{b-1},\hat{w}_{1,b-1},j_{b-1})$ and $\dv y[b-1]$ are jointly typical for some $j_{b-1} \in [1,J]$. One can show that the decoding error in this step is small for sufficiently large $n$ if \eqref{Constraint__On__IndividualRate} is true.

\textbf{Fourier-Motzkin Elimination:} From the above, we get that the error probability is small provided that $n$ is large and
\begin{subequations}
\begin{align}
R_0 &< I(X_2;Y)\\
\hat{R} &< I(V;Y|X_2)+R_0\\
\hat{R} &> I(V;S|X_2)\\
R_1 &\leq I(U;Y|V,X_2)-I(U;S|V,X_2)\\
R_c + R_1 + \hat{R} &\leq I(U,V,X_2;Y)-I(U;S|V,X_2).
\end{align}
\label{RateRegion__Step1}
\end{subequations}
Applying Fourier-Motzkin Elimination (FME) to project out  $\hat{R}$ and $R_0$ from \eqref{RateRegion__Step1}, we get
\begin{subequations}
\begin{align}
0 &\leq I(V,X_2;Y)-I(V,X_2;S)\\
R_1 &\leq I(U;Y|V,X_2)-I(U;S|V,X_2)\\
R_c + R_1 &\leq I(U,V,X_2;Y)-I(U,V,X_2;S).
\end{align}
\label{RateRegion__Step2}
\end{subequations}

\subsubsection{Bounds on $|\mc V|$ and $|\mc U|$}

It remains to show that the rate pair \eqref{EquivalentCharacterizationCapacityRegionDiscreteMemorylessChannel} is not altered if one restricts the random variables $V$ and $U$ to have their alphabet sizes limited as indicated in \eqref{BoundsOnCardinalityOfAuxiliaryRandonVariables__EquivalentCharacterizationCapacityRegion__DiscreteMemorylessChannel}. This is done by a standard application of the support lemma \cite[p. 310]{CK81}, essentially by following the lines in the proof of Theorem~\ref{Theorem__CapacityRegionDiscreteMemorylessChannel} in Appendix~\ref{appendixTheorem__CapacityRegionDiscreteMemorylessChannel} and noticing that, this time, because of the additional nonnegativity constraint, one more functional needs to be preserved in bounding the cardinality of $V$,
\begin{align}
I_{\mu}(V,X_2;Y)-I_{\mu}(V,X_2;S) = H_{\mu}(Y)-H_{\mu}(S)+H_{\mu}(X_2,S|V)-H_{\mu}(X_2,Y|V).\end{align}

This concludes the proof of Corollary~\ref{Corollary__EquivalentCharacterizationCapacityRegionDiscreteMemorylessChannel}.
%%%%%%%%%%%%%%%%%%%%%%%%%%%%%%%%%%%%%%%%%%%%%%%%%%%%%%%%%%%%%%%%%%%%%%%%%%%%%%%%%%%%%%%%
\renewcommand{\theequation}{E-\arabic{equation}}
\setcounter{equation}{0}  % reset counter
\subsection{Proof of Theorem~\ref{Theorem__AlternativeOuterBoundDiscreteMemorylessChannel}}\label{appendixTheorem__AlternativeOuterBoundDiscreteMemorylessChannel}

We prove that for any $(M_c,M_1,n,\epsilon)$ code consisting of a mapping $\phi_1:\mc W_c{\times}\mc W_1{\times}\mc S^n \longrightarrow \mc X^n_1$ at Encoder 1, a sequence of mappings $\phi_{2,i}: \mc W_c{\times}\mc S^{i-1} \longrightarrow \mc X_2$, $i=1,\hdots,n$, at Encoder 2, and a mapping $\psi : \mc Y^n \longrightarrow \mc W_c{\times}\mc W_1$ at the decoder with average error probability $P_e^n \rightarrow 0$ as $n \rightarrow 0$ and rates $R_c=n^{-1}\log_2M_c$ and $R_1=n^{-1}\log_2M_1$, the rate pair $(R_c,R_1)$ must satisfy \eqref{AlternativeOuterBoundDiscreteMemorylessChannel}.

 \noindent Fix $n$ and consider a given code of block length $n$. The joint probability mass function on $\mc W_c{\times}\mc W_1{\times}\mc S^n{\times}\mc X^n_1{\times}\mc X^n_2{\times}\mc Y^n$ is given by
 \begin{align}
 P(w_c,w_1,s^n,x^n_1,x^n_2,y^n)=P(w_c,w_1)\prod_{i=1}^nP(s_i)P(x_{1i}|w_c,w_1,s^n)P(x_{2i}|w_c,s^{i-1})P(y_i|x_{1i},x_{2i},s_i),
 \end{align}
 where, $P(x_{1i}|w_c,w_1,s^n)$ is equal $1$ if $x_{1i}=f_1(w_c,w_1,s^n)$ and $0$ otherwise; and $P(x_{2i}|w_c,s^{i-1})$ is equal $1$ if $x_{2i}=f_2(w_c,s^{i-1})$ and $0$ otherwise.

\noindent The proof of the bound on $R_1$ follows trivially by revealing the state $S^n$ to the decoder.

\noindent The proof of the bound on the sum rate $R_c+R_1$ is as follows. The decoder map $\psi$ recovers $(W_c,W_1)$ from $Y^n$ with vanishing average error probability. By Fano's inequality, we have
\begin{align}
H(W_c,W_1|Y^n) \leq n\epsilon_n,
\end{align}
where $\epsilon_n \rightarrow 0$ as $P_e^n \rightarrow 0$.

{\allowdisplaybreaks
\begin{align}
n(R_c+R_1) &= H(W_c,W_1)\nonumber\\
        & = I(W_c,W_1;Y^n)+H(W_c,W_1|Y^n)\nonumber\\
	& \stackrel{(a)}{\leq}I(W_c,W_1;Y^n)+n\epsilon_n\nonumber\\
	& = I(W_c,W_1,S^n;Y^n)-I(S^n;Y^n|W_c,W_1)+n\epsilon_n\nonumber\\
	& = \Big(\sum_{i=1}^{n}I(W_c,W_1,S^n;Y_i|Y^{i-1})\Big)-H(S^n|W_c,W_1)+H(S^n|W_c,W_1,Y^n)+n\epsilon_n\nonumber\\
	& \stackrel{(b)}{=} \sum_{i=1}^{n} H(Y_i|Y^{i-1})-H(Y_i|W_c,W_1,S^n,Y^{i-1})-H(S_i)+H(S_i|W_c,W_1,Y^n,S^{i-1})+n\epsilon_n\nonumber\\
	& \stackrel{(c)}{\leq} \sum_{i=1}^{n} H(Y_i)-H(Y_i|X_{1,i},X_{2,i},S_i)- H(S_i)+H(S_i|W_c,W_1,Y^n,S^{i-1},X_{2,i})+n\epsilon_n\nonumber\\
	& \stackrel{(d)}{\leq} \sum_{i=1}^{n} I(X_{1,i},X_{2,i},S_i;Y_i)-H(S_i)+H(S_i|X_{2,i},Y_i)+n\epsilon_n\nonumber\\
	& = \sum_{i=1}^{n} I(X_{1,i},X_{2,i},S_i;Y_i)-I(S_i;X_{2,i},Y_i) +n\epsilon_n\nonumber\\
	& = \sum_{i=1}^{n} I(X_{1,i},X_{2,i};Y_i|S_i)-I(S_i;X_{2,i}|Y_i)+n\epsilon_n,
\label{BoundOnSumRate}
\end{align}}
where: $(a)$ follows from Fano's inequality; $(b)$ follows from the fact that the state $S^n$ is i.i.d. and is independent of the messages; $(c)$ follows from $(W_c,W_1,S^n,Y^{i-1}) \leftrightarrow (X_{1,i},X_{2,i},S_i) \leftrightarrow Y_i$, and the fact that $X_{2,i}$ is a deterministic function of $(W_c,S^{i-1})$; and $(d)$ follows from the fact that conditioning reduces entropy.

Finally, we obtain the desired bound from \eqref{BoundOnSumRate} by standard single-letterization \cite{CK81}.

%%%%%%%%%%%%%%%%%%%%%%%%%%%%%%%%%%%%%%%%%%%%%%%%%%%%%%%%%%%%%%%%%%%%%%%%%%%%%%%%%%%%%%%%
\renewcommand{\theequation}{F-\arabic{equation}}
\setcounter{equation}{0}  % reset counter
\subsection{Proof of Corollary~\ref{Corollary__CommonMessageCapacity}}\label{appendixCorollary__CommonMessageCapacity}

Relaxing the constraint on $R_1$ in Theorem~\ref{Theorem__CapacityRegionDiscreteMemorylessChannel}, we obtain
\begin{align}
C &= \max I(U,V,X_2;Y)-I(U,V,X_2;S)
\end{align}
where the maximization is over joint measures $P_{S,U,V,X_1,X_2,Y}$  of the form
\begin{align}
P_{S,U,V,X_1,X_2,Y} &= Q_SP_{X_2}P_{V|S,X_2}P_{U,X_1|S,V,X_2}.
\end{align}

\noindent The corollary then follows by substituting $K=(U,V)$, and noticing that the distribution on $(S,K,X_1,X_2,Y)$ is given by
\begin{align}
P_{S,K,X_1,X_2,Y} &= P_{S,U,V,X_1,X_2,Y}\\
                  &= Q_SP_{X_2}P_{V|S,X_2}P_{U,X_1|S,V,X_2}\\
                  &= Q_SP_{X_2}P_{U,V|S,X_2}P_{X_1|S,U,V,X_2}\\
                  &= Q_SP_{X_2}P_{K|S,X_2}P_{X_1|S,K,X_2}.
\end{align}

%%%%%%%%%%%%%%%%%%%%%%%%%%%%%%%%%%%%%%%%%%%%%%%%%%%%%%%%%%%%%%%%%%%%%%%%%%%%%%%%%%%%%%%%
\renewcommand{\theequation}{G-\arabic{equation}}
\setcounter{equation}{0}  % reset counter
\subsection{Proof of Theorem~\ref{Theorem__CapacityRegionMemorylessGaussianChannel}}\label{appendixTheorem__CapacityRegionMemorylessGaussianChannel}

\subsubsection{Direct Part} The achievability follows by ignoring the strictly causal part of the state at Encoder 2, and using the generalized dirty paper coding scheme of \cite[Theorem 7]{SBSV07a}.

\subsubsection{Converse Part} For the converse part, we use the outer bound of Theorem~\ref{Theorem__AlternativeOuterBoundDiscreteMemorylessChannel} for the discrete MAC which can be readily extended to memoryless channels with discrete time and continuous alphabets using standard techniques \cite{G68}. Then, we obtain an outer bound on the capacity region of the Gaussian MAC in terms of the closure of the convex hull of the set of rate pairs $(R_c,R_1)$ satisfying
\begin{align}
R_1 \: &\leq \: I(X_1;Y|S,X_2), \nonumber\\
R_c+ R_1 \: &\leq \: I(X_1,X_2;Y|S)-I(X_2;S|Y),
\label{OuterBoundGaussianChannel}
\end{align}
for some probability distribution of the form $P_{S,X_1,X_2,Y}=Q_SP_{X_2}P_{X_1|X_2,S}W_{Y|X_1,X_2,S}$ such that $\mathbb{E}[X^2_1] \leq P_1$ and $\mathbb{E}[X^2_2] \leq P_2$. The rest of the converse proof follows by reasoning and using algebra similar to in the proofs of \cite[Theorem 7]{SBSV07a} and \cite[Theorem 4]{ZKLV10}, and is omitted for brevity.

\bibliographystyle{IEEEtran}
\bibliography{Draft__InitialSubmission__OneColumn}
\end{document}